\documentclass[letterpaper,10pt,conference,romanappendices]{ieeeconf}

\IEEEoverridecommandlockouts
\usepackage{amsmath,amssymb,theorem}
\usepackage[mathscr]{euscript}
\usepackage{graphicx}
\usepackage{subfig}
\usepackage{algorithm,algorithmic}
\usepackage{psfrag}
\usepackage{cite}
\usepackage{tikz}
\usetikzlibrary{shadows.blur}
\usetikzlibrary{arrows} 
\tikzset{
    >=stealth',
    pil/.style={
           ->,
           thick,
           shorten <=2pt,
           shorten >=2pt,}
}
\usepackage{bbold}
\usepackage{color,xspace}
\newtheorem{theorem}{Theorem}[section]
\newtheorem{lemma}[theorem]{Lemma}

\newtheorem{proposition}[theorem]{Proposition}
\newtheorem{problem}[theorem]{Problem}

\newcommand{\real}{{\mathbb{R}}}

\newcommand{\realnonnegative}{\mathbb{R}_{\ge 0}}
\newcommand{\integernonnegative}{\mathbb{Z}_{\ge 0}}
\newcommand{\integerpositive}{\mathbb{Z}_{> 0}}

\newcommand{\DD}{{\mathcal{D}}}

\newcommand{\PP}{{\mathcal{P}}}

\newcommand{\II}{{\mathcal{I}}}
\newcommand{\MM}{{\mathcal{M}}}
\newcommand{\NN}{{\mathcal{N}}}

\newcommand{\UU}{{\mathcal{U}}}

\newcommand{\HH}{{\mathcal{H}}}

\newcommand{\KK}{\mathcal{K}}
\newcommand{\RR}{\mathcal{R}}

\newcommand{\ragents}{{\mathcal{A}}}

\newcommand{\bound}{\operatorname{bnd}}
\newcommand{\halfspace}{{H}}

\newcommand{\VV}{{\mathcal{V}}}

\newcommand{\cm}[1]{\operatorname{cntr}(#1)}

\newcommand{\unit}[1]{\operatorname{unit}(#1)}
\newcommand{\tbb}{\operatorname{tbb}}
\newcommand{\loc}[1]{\operatorname{loc}(#1)}

\renewcommand{\tilde}{\widetilde}

\newcommand{\proj}{\operatorname{pr}}

\newcommand{\vmax}{v_{\text{max}}}
\newcommand{\timestep}{\Delta t}

\renewcommand{\epsilon}{\varepsilon}
\newcommand{\diam}{\operatorname{diam}}

\newcommand{\gV}{\text{g}V}

\newcommand{\dgV}{\text{dg}V}

\newcommand{\gW}{\text{g}W}
\newcommand{\dgW}{\text{dg}W}

\newcommand{\onevec}{\mathbf{1}}
\newcommand{\dataspace}{(S \times \realnonnegative)^{n^2}}

\newcommand{\algostep}[1]{{\small\texttt{#1:}}\xspace}

\newcommand{\until}[1]{\{1,\dots, #1\}}

\newcommand{\map}[3]{#1: #2 \rightarrow #3}
\newcommand{\svmap}[3]{#1: #2 \rightrightarrows #3}
\newcommand{\setdef}[2]{\{#1 \; | \; #2\}}

\newcommand{\cball}[2]{\overline{B}(#1,#2)}

\newcommand{\csegment}[2]{[#1,#2]}

\newcommand{\TwoNorm}[1]{\|#1\|}

\renewcommand{\hat}{\widehat}

\newcommand{\algoMotionControl}{\texttt{motion control law}\xspace}
\newcommand{\algoOneStep}{\texttt{one-step-ahead update policy}\xspace}
\newcommand{\algoMultiStep}{\texttt{multiple-steps-ahead update  policy}\xspace}
\newcommand{\algoVoronoi}{\texttt{Dominant cell computation}\xspace}

\newcommand{\algoFull}{\texttt{$k$-order self-triggered centroid
    algorithm}\xspace}

\newcommand{\oprocendsymbol}{\hbox{$\bullet$}}
\newcommand{\oprocend}{\relax\ifmmode\else\unskip\hfill\fi\oprocendsymbol}

\renewcommand{\theenumi}{(\roman{enumi})}

\renewcommand{\cm}[1]{C_{#1}}

\begin{document}



\title{Self-triggered distributed $k$-order coverage control}

\author{Daniel Tabatabai 
	\and Mohanad Ajina \and Cameron Nowzari \thanks{The authors are
    with the Department of Electrical and Computer Engineering, George Mason University, Fairfax, VA 22030, USA, {\tt\small \{dtabatab,majina,cnowzari\}}@gmu.edu}}
\date{\today}

\maketitle

\begin{abstract} 
A $k$-order coverage control problem is studied where a network of agents must deploy over a desired area. The objective is to deploy all the agents in a decentralized manner such that a certain coverage performance metric of the network is maximized. Unlike many prior works that consider multi-agent deployment, we explicitly consider applications where more than one agent may be required to service an event that randomly occurs anywhere in the domain. The proposed method ensures the distributed agents autonomously cover the area while simultaneously relaxing the requirement of constant communication among the agents. In order to achieve the stated goals, a self-triggered coordination method is developed that both determines how agents should move without having to continuously acquire information from other agents, as well as exactly when to communicate and acquire new information. Through analysis, the proposed strategy is shown to provide asymptotic convergence similar to that of continuous or periodic methods. Simulation results demonstrate that the proposed method can reduce the number of messages exchanged as well as the amount of communication power necessary to accomplish the deployment task.
\end{abstract}

\section{Introduction}\label{se:intro}
This paper studies a multi-agent coordination problem where a network of agents perform a deployment task to statically position themselves over a desired area. For example, a mobile sensor network where it is required to deploy sensors to positions that will maximize total coverage of the desired sensing environment. This is commonly referred to as coverage control. Specific applications include topics such as environmental monitoring \cite{curtin1993autonomous, lu2017cooperative}, survelliance \cite{peters2017coverage}, data collection \cite{rybski2000enlisting, zhong2011distributed}, and search and rescue \cite{macwan2011optimal}. More specifically, we consider a generalization of the coverage problem that extends to scenarios where more than one agent may be required to overlap a region in the coverage area. {\color{black}This is referred to as $k$-order coverage control~\cite{gallais2007areacoverage, wang2009energy, yu2017kcoverage}, where $k > 1$ agents must overlap coverage of the same point $q$.}
Our contributions focus on the development of coordination strategies that will reduce the amount of communication necessary between agents while performing the deployment task. This is accomplished by the design of a self-triggered algorithm where agents autonomously decide when they require information from other agents in the network.

{\color{black}
With respect to coverage control, the majority of previous research has focused on scenerios where an individual agent is capable of servicing events that occur in the agent's respective region of responsibility without the assistance of other agents. As an example, consider a monitoring application where a wirless sensor network must monitor the environment. If a random event occurs in the vicinity of a particular sensor then that sensor has the ability to measure and capture the event independent of other sensors in the network. However, various applications exists where agents do not possess the capability to capture or respond to events independently. These applications require multiple agents to work collectively in order to service events. One example of this type of application is that of Time Difference of Arrival (TDOA) localization \cite{gustafsson2003positioning, gardner1992signal, mellen2003closed} where the requirement is that three or more sensors that are located at different positions must measure the same event. Another example is emergency response vehicles where two or more vehicles may be required to respond to a particular event, such as a fire or burglary. In other scenarios, two or more agents may not be necessary for event handling, but the application may require redundant agents to overlap areas for fail-safe purposes.
}
\paragraph*{Literature review}
The topic of multi-agent coverage control has been studied by a number of authors in the past including the seminal work~\cite{cortes2004coverage} {\color{black}where coverage control based on agents moving to the centroids of a Voronoi partition was introduced.}
In~\cite{poduri2004constrained}, the authors consider the coverage control problem where each node is constrained to have $m$ neighboring nodes. The authors use an approach based on vector potential fields where each node acts as a repelling force in order to maximize coverage and acts as an attracting force in order to satisfy the $m$ 
neighbor constraint. In \cite{pimenta2008sensing}, the authors consider heterogeneous and non-point source nodes as well as non-convex enviroments. In \cite{schwager2007decentralized}, the authors study the problem in the context of using sensor measurements to estimate regions of importance in the mission space thus driving nodes to concentrate in these areas. {\color{black}Common to all the above mentioned works is the fact that they study the coverage control problem in terms of a first-order coverage problem where each agent is solely responsible for covering a sub-region of the mission space. As previously mentioned, the interest of our work is the generalized $k$-order coverage control problem where multiple agents overlap coverage of sub-regions in the mission space.}

The $k$-order coverage control problem was studied in \cite{jiang2015higher,DBLP:journals/corr/JiangSAL17,JIANG201927} where a method using higher-order Voronoi partitions was proposed. The authors present a method for deploying agents over a bounded area when more than one agent must have overlapping coverage of the same point. However, to realize the proposed contol law in \cite{JIANG201927}, it is assumed that continuous communication between agents is achievable. {\color{black}For many real-world systems,} 
continuous communication is not feasible and periodic solutions can be resource inefficient and may not be neccessary. As alternatives to continuous and periodic solutions, self-triggered and event-triggered approaches have been proposed in the literature to handle similiar problems in networked systems~\cite{heemels2012introduction, dimarogonas2010distributed, dimarogonas2009event, mazo2008event, wang2009energy}. For self- and event-triggered solutions, the exact time at which agents perform actions, e.g. wirelessly communicate or update a control signal, is autonomously decided by the agents rather than occurring at periodic time intervals. 


In~\cite{nowzari2012self,CN-JC-GJP:15-acc}, the concepts of self-triggered control was applied to the case of first-order optimal deployment. In our current work, we extend the self-triggered centroid algorithm presented in~\cite{nowzari2012self} by considering the higher-order coverage control problem studied in~\cite{jiang2015higher,DBLP:journals/corr/JiangSAL17,JIANG201927} and develop a self-triggered coordination strategy to relax the synchronous, periodic communication requirement {\color{black}while guaranteeing that each agent moves such that it does not contribute negatively to the task.} 
\paragraph*{Statement of contributions}

The main contribution of this work is the development of a distributed self-triggered control strategy that deploys a set of agents to static locations in a convex area in order to achieve~$k$-order optimal coverage. Our solution relaxes the need for continuous or periodic communication among agents as is done in prior works~\cite{JIANG201927}. More specifically, our algorithm is comprised of two major sub-components. The first being an update decision policy where each agent decides when to acquire new information from neighboring agents through a wireless communication network. The decision to comunicate is based on the level of uncertainty each agent has accumulated over time. This uncertainty is due to not having up-to-date information that results from the lack of communication with other agents. {\color{black}We extend the notion of uncertain spatial partitioning~\cite{evans2008guaranteed, jooyandeh2009uncertain, cheng2010uv} used for optimal deployment in~\cite{nowzari2012self} by the use of $k$-order guaranteed and dual-guaranteed Voronoi partitions.}
The second major sub-component is a motion control law that determines how agents should move given possibly outdated information about the location of other agents in the network. Each agent determines a motion plan that is guaranteed to contribute positively to the higher-order deployment task.

\paragraph*{Organization} 
Section \ref{se:pre} outlines some important notions from computational geometry. Section \ref{se:statement} formally presents the problem statement. Section \ref{se:partition} {\color{black}formulates the concepts of $k$-order guaranteed and $k$-order dual-guaranteed Voronoi partitions.} Section \ref{se:design} presents the algorithm design. In section \ref{se:convergence} convergence analysis of the algorithm is discussed. 
Section \ref{se:simulations} presents simulation results and section \ref{se:conclusions} assimilates the conclusions.

\section{Preliminaries}\label{se:pre}
Let $\realnonnegative$ and $\integernonnegative$ be the set of non-negative real, integer values respectively. With the Euclidean norm defined by $\TwoNorm{\cdot}$

\subsection{Basic geometric notions}
We denote by $\csegment{p}{q} \subset \real^d$ the closed segment with
extreme points $p$ and $q \in \real^d$.  Let
$\map{\phi}{\real^d}{\realnonnegative}$ be a bounded measurable
function that we term \emph{density}. For $S \subset \real^d$, the
\emph{mass} and \emph{center of mass} of $S$ with respect to $\phi$
are
\begin{align*}
  M_S =&\; \int_S \phi (q) dq, \hspace*{5 mm}
  \cm{S} = \frac{1}{M_S} \int_S q \phi (q) dq .
\end{align*}
Let $s_1,s_2,...,s_n$ be $n$ subsets of $S$ and $\{s_1,s_2,...,s_n\}$ be a partition of $S$ then  mass and center of mass with respect to $\phi$ and the $n$ partitions, 
\begin{align*}
M_S = \sum_{i=1}^n M_{s_i}, \hspace*{5 mm}
\cm{S} = \frac{\sum_{i=1}^n M_{s_i}\cm{s_i}}{\sum_{i=1}^n M_{s_i}}
\end{align*}
The \emph{circumcenter} $\text{cc}_s$ of a bounded set $S \subset \real^d$ is the center of a closed ball of minimum radius that contains $S$. The \emph{circumradius} $\text{cr}_s$ of $S$ is the radius of this ball. The diameter of $S$ is $\text{diam}(S) = \max_{p,q \in S} \TwoNorm{p-q}$.

Given $v\in \real^d \setminus \{0\}$, let $\unit{v}$ be the unit
vector in the direction of $v$.  Given a convex set $S \subset
\real^d$ and $p \in \real^d$, let $\proj_S (p)$ denote the orthogonal
projection of $p$ onto $S$, i.e., $\proj_S (p)$ is the point in $S$
closest to $p$.  The \emph{to-ball-boundary} map $\tbb : (\real^d
\times \realnonnegative)^2 \rightarrow \real^d$ takes $(p, \delta, q,r
)$ to
\begin{align*}
  \begin{cases}
    p + \delta \unit{q-p} & \text{if } \TwoNorm{p - \proj_{\cball{q}{r}}
      (p)} \ge \delta ,
    \\
    \proj_{\cball{q}{r}} & \text{if } \TwoNorm{p - \proj_{\cball{q}{r}}
      (p)} \leq \delta .
  \end{cases}
\end{align*}
Figure \ref{fig:tbb} illustrates the action of $\tbb$.

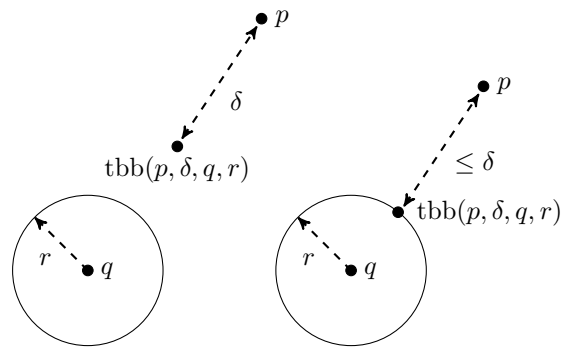
\begin{figure}[htb]
  \centering
  \begin{tikzpicture}
  \draw[white] (0,0) -- (0,5) -- (8,5) -- (8,0) -- (0,0);
 
  \filldraw[fill=white, draw=black] (1.5,1.25) circle (1);
  
  \draw [dashed, thick, <->] (2.75, 3)--(3.75, 4.5);
  \node[below right] at (3.27, 3.75) {$\delta$};
  \filldraw[black] (2.69, 2.9) circle (2pt);
  \node[below] at (2.69, 2.9) {$\tbb(p,\delta,q,r)$};
  \filldraw[black] (3.81, 4.6) circle (2pt);
  \node[right] at (3.87, 4.6) {$p$};  
  
  \filldraw[black] (1.5,1.25) circle (2pt);
  \node[right] at (1.55,1.25) {$q$};
  \draw[dashed, thick, ->] (1.5,1.25) -- (0.7929, 1.9571);
  \node[below left] at (1.1465, 1.6035) {$r$};
  
  \filldraw[fill=white, draw=black] (5,1.25) circle (1);

  \filldraw[black] (5, 1.25) circle (2pt);
  \node[right] at (5.05, 1.25) {$q$};
  \draw[dashed, thick, ->] (5, 1.25) -- (4.2929, 1.9571);
  \node[below left] at (4.6465, 1.6035) {$r$};
  
  \draw [dashed, thick, <->] (5.7, 2.1)--(6.7, 3.6);
  \node[below right] at (6.25, 2.95) {$\leq \delta$};
  \filldraw[black] (5.625, 2.025) circle (2pt);
  \node[right] at (5.75, 2.025) {$\tbb(p,\delta,q,r)$};
  \filldraw[black] (6.76, 3.7) circle (2pt);
  \node[right] at (6.82, 3.7) {$p$};
  
  \end{tikzpicture}
  \caption{Graphical representation of the action of $\tbb$ when (a)
    $\TwoNorm{p - \proj_{\cball{q}{r}} (p)} > \delta$ and (b)
    $\TwoNorm{p - \proj_{\cball{q}{r}} (p)} \leq
    \delta$.}\label{fig:tbb}
\end{figure}

We denote by $\cball{p}{r}$ the closed
ball centered at $p \in S$ with radius $r$ and by $\halfspace_{po} =
\setdef{q\in\real^d}{ \| q - p \| \le \| q- o \|}$ the closed
halfspace determined by $p, o \in \real^d$ that contains $p$.

\subsection{1-order Voronoi partitions}
The methods developed in this work rely heavily on the concept of Voronoi partitioning \cite{senechal1993spatial}. In the following sub-sections a brief discussion of 1-order Voronoi partitions is presented.
Let $S$ be a convex polygon in $\real^2$ and $P=(p_1,\dots,p_n)$ 
be the location of $n$ agents. A \emph{partition} of $S$ is a collection of $n$ polygons $\KK=\{K_1,\dots,K_n\}$ with disjoint interiors whose union is~$S$.
The \emph{Voronoi partition} $\VV(P)=\{V_1,\dots,V_n\}$ of $S$
generated by the points $P = (p_1,\dots,p_n)$ is
\begin{equation*}
  V_i = \setdef{q\in S}{\| q - p_i\| \leq \|q - p_j \| \, , \; \forall
    j\neq i}. 
\end{equation*}
Intuitively, the Voronoi cell $V_i$ represents all the points that are closer to the agent at position $p_i$ than to any of the other agents in the network.
When the Voronoi regions $V_i$ and $V_j$\ are adjacent (i.e., they
share an edge),
$p_i$ is called a \emph{(Voronoi) neighbor} of $p_j$
(and vice versa). $P=(p_1,\dots,p_n)$ is a \emph{centroidal Voronoi configuration} if it satisfies that $p_i = \cm{V_i}$, for all $i \in \until{n}$.

\section{Problem statement}\label{se:statement}


\subsection{k-order Voronoi partitions}
\begin{figure*}
\centering%
\subfloat[1-order Voronoi diagram]{\includegraphics[scale=0.21]{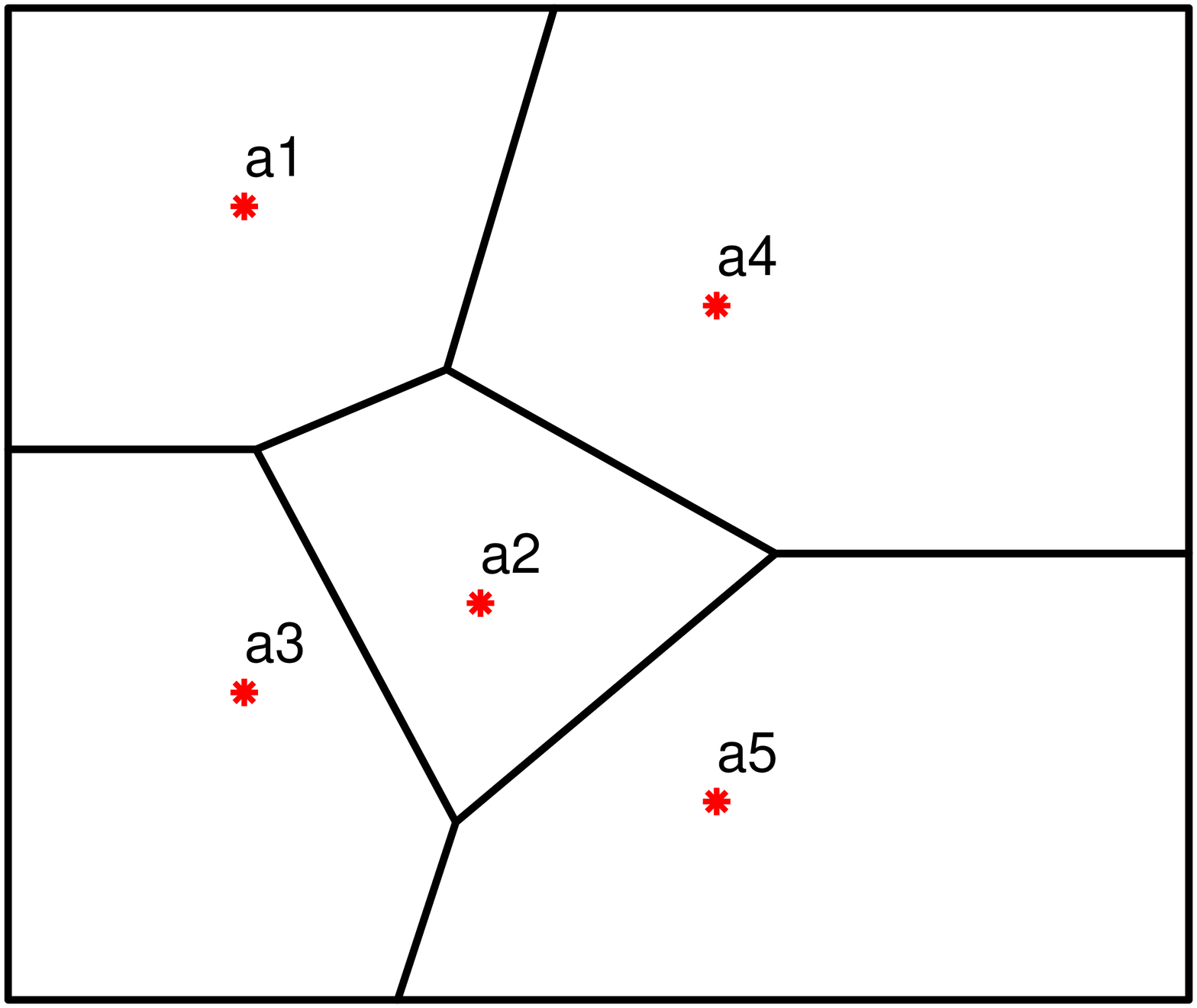} \label{fig:1order-ex}}%
\subfloat[2-order Voronoi diagram]{\includegraphics[scale=0.21]{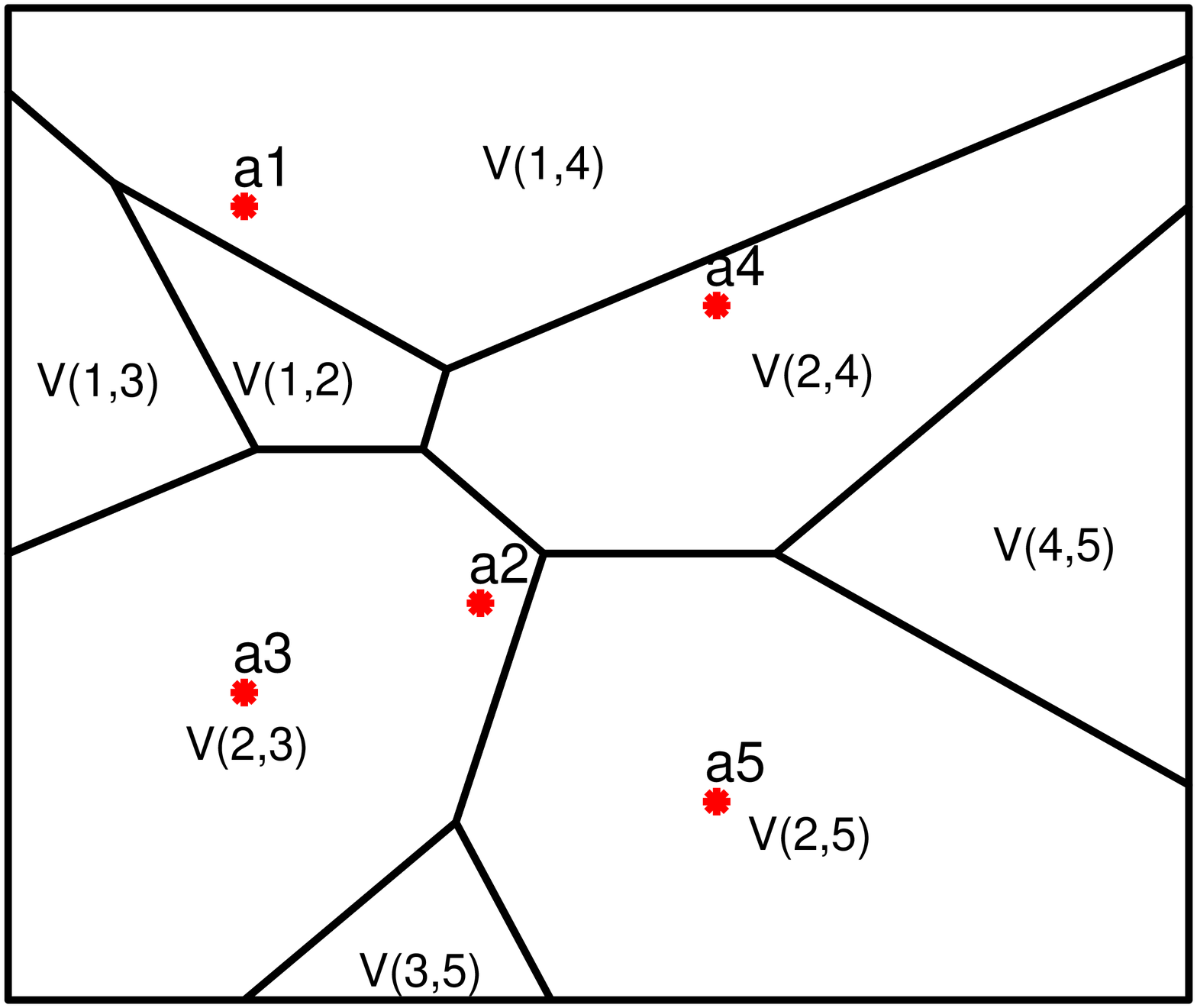} \label{fig:2order-ex}}%
\caption{Example of 1-order Voronoi diagram (left) and 2-order Voronoi diagram (right). In the 1-order case, cells can be represented by the agent that covers the cell. In the case of 2-order partitions, two agents share coverage over a cell. The cells are represented by two agents that cover the particular partition. Note that not all agents share cells. In particular, agents~$(1,5)$ and agents~$(3,4)$ do not share cells, i.e. $V_{(1,5)}=\emptyset$ and $V_{(3,4)}=\emptyset$}%
\label{fig:vorono-ex}%
\end{figure*}
%

{\color{black}Intuitively, a $k$-order Voronoi cell represents all the points that are closer to $k$~agents located at positions $\{p_{i_1},p_{i_2},...,p_{i_k}\} = \PP_I$ than to any of the other agents in the network. A $k$-order Voronoi partition would then be a collection of all the $k$-order Voronoi cells. In the first-order case, the space is partitioned into ~$n$ cells such that each agent is closest to every point in their cell than any of the other agents. However, in the case of a second-order partition the space is partitioned into~$\frac{n(n-1)}{2}$ cells and there are many cells that can be empty because certain agents may not share overlapping responsibility for any points in the space. The difference between the first-order and second-order case can be seen in figure \ref{fig:vorono-ex}. In figure \ref{fig:1order-ex} an example of a first-order partition is illustrated for five agents. Note that each agent is enclosed in their respective cells and there is exactly five cells, one per agent. Figure \ref{fig:2order-ex} illustrates the second-order partition for the same five agents. Note that in figure \ref{fig:2order-ex} the number of cells are greater than the number of agents. Also note that some cells contain multiple agents while other cells do not contain any agents at all. Furthermore, some agent combinations are not associated to a cell at all, e.g. $V_(1,5)$ and $V_(3,4)$. This is due to the fact that these agents do not share points in the space that are mutually closer to them combined than to any of the other agents. A more formal definition of the $k$-order Voronoi partition follows.

Let $S \subset \real^2$ be a convex polygon in a $2$-dimensional space.
Let $\ragents = \{1,\dots, n\}$ be a finite set of integers representing the agents in a $n$-agent network.
Let $\PP = \{p_1, \dots, p_n\} \in S$ be the set of positions of the agents~$\ragents$ in the domain~$S$.
Let $I = (i_1, \dots, i_k) \in \II$ be a $k$-tuple element of the set $\II$ where $\II = \setdef{(i_1, \dots, i_k) \in \ragents^k}{i_1 < \dots < i_k}$ is the set of $k$-tuples in $\ragents^k$ that do not repeat, for example $(i_1,i_2,i_3) \in \II$ , but $(i_3,i_2,I_1) \notin \II$. Let $\PP_I = \{p_{i_1}, \dots p_{i_k} \} \subset \PP$ be the subset of agent positions corresponding to the agents $(i_1, \dots, i_k) \in \II$. The collection of all elements $I \in \II$ that include a particular agent~$i$ is denoted by $\II^i$ where $\II^i = \setdef{I \in \II}{\forall~1\leq\alpha\leq k, i = I_\alpha}$
A $k$-order~\emph{partition} of $S$ is a collection of $m$ polygons $\RR=\{R_1,\dots,R_n\}$ with disjoint interiors whose union is~$S$ and where an element in $\RR$ is associated with $k$-agents.}

The $k$-order Voronoi partition of a convex polygon $S$ can be defined as follows. Given a set of agents $\ragents$ with positions $\PP$. For $k < n$, let $\PP_I \subset \PP$ with $|\PP_I|=k$. Then the $k$-order Voronoi region associated with agents $I=(i_1,\dots,i_k)$
with generating sites $P_I=(p_{i_1},\dots,p_{i_k})$ is defined as 
\[
V_I = \setdef{q \in S}{\TwoNorm{q-p} \leq \TwoNorm{q-p'},
 ~\forall~ p \in \PP_I, ~\forall~ p' \in \PP \setminus \PP_I}.
\]
For example, if $k=2$ then $\PP_I = \{p_i,p_j\}$ and the second-order Voronoi cell for agents $(i,j)\in\mathcal{I}$ becomes,
\begin{align*}
V_{ij} = \{q\in S ~|~ & \TwoNorm{q-p_i} \leq \TwoNorm{q-p'}, \\ 
& \TwoNorm{q-p_j} \leq \TwoNorm{q-p'}, ~\forall~ p' \in \PP \setminus \{p_i,p_j\} \}
\end{align*}
For every point $q$ in $V_{I}$, the distance from $q$ to any agent position in $\PP_I$ is less than or at most equal to the distance from $q$ to all other agent positions not in $\PP_I$. For $k=2$, the second-order Voronoi partition with $I=(i,j)$ and $\PP_I=\{p_i,p_j\}$ would mean that the two agents $i$ and $j$ are closer to or at most as close to all the points in $V_{ij}$ than any of the other agents $\ragents \setminus \{i,j\}$. An alternative interpretation would be that the agents $i$ and $j$ are considered responsible for the region defined by $V_{ij}$. 

Combining all $k$-order Voronoi regions in $S$, the $k$-order Voronoi partition of the environment $S$ becomes
$\VV(P) = \big\{V_I\big\}_{I\in\mathcal{I}}$. The environment $S$ can be considered as the union of all $k$-order Voronoi cells $S = \bigcup_{I\in\mathcal{I}}V_I$. Figure~\ref{fig:vorono-ex} presents an example of the difference between a first-order (\ref{fig:1order-ex}) and second-order (\ref{fig:2order-ex}) Voronoi partition for five agents.
For any agent~$i$ with position $p_i \in \PP$, there can be multiple sets $\PP_I \subset \PP$ that contain $p_i$ meaning that an agent $i$ located at position $p_i$ can be responsible for multiple $k$-order Voronoi cells. The collection of $k$-order Voronoi cells associated with agent~$i$ is given by $\VV^i = \{V_I\}_{I\in\II^i}$. All $k$-order cells associated with agent~$i$ can be combined to form a single region of $S$ that agent~$i$ is responsible for and this cell is referred to as the dominant region of agent~$i$.
The dominant region for agent~$i$ is be defined by
\[
W_i = \bigcup_{I \in \II^i} V_I.
\]
The dominant cell $W_i$ represents the region of $S$ that agent $i$ is responsible for covering. Note that the first-order cell $V_i$ and the $k$-order cell $V_{I}$ are not equivalent, but both are convex. However, the dominant cell $W_i$ may not be convex. The $k$-order neighbors of agent $i$ is denoted by $\NN_i$. For a $k$-order Voronoi partition, $P=(p_1,\dots,p_n)$ is a {\color{black}\emph{centroidal~$k$- order Voronoi configuration}} if it satisfies $p_i = \cm{W_i}$, for all $i \in \until{n}$. Next, optimal deployment for $k$-order Voronoi partitioning is discussed.

\subsection{Objective for higher-order coverage}\label{se:objective-higher-order}
The interest is in applications where~$k > 1$ agents are required to service an event occuring at a random point $q \in S$. This is in contrast to the 1-order problem where for any point $q \in S$ only one agent is responsible. In order to optimally deploy agents throughout the mission space, an objective function for the higher order deployment problem must be defined. For the 1-order case, from \cite{cortes2004coverage}, the objective function in terms of Voronoi partitions is defined as 
\begin{equation} \label{eq:HH_VV}
\HH (P) = \sum_{i=1}^{n} \int_{V_i}\TwoNorm{q-p_i}^2 \phi(q) dq
\end{equation}
The objective here is to minimize the distance from agent $i$'s position $p_i$ to all points $q \in V_i$. Taking advantage of the parallel axis theorem, $\HH (P)$ may be expressed as,
\begin{equation} \label{eq:HH_VV_polar}
\HH (P) = \sum_{i=1}^{n}J_{V_i,C_{V_i}} + \sum_{i=1}^{n}M_{V_i}\TwoNorm{p_i - C_{V_i}}^2
\end{equation}
where $J_{V_i,C_{V_i}}$ is the polar moment of inertia of the 1-order Voronoi cell $V_i$ centered at the centroid $C_{V_i}$. Taking the partial derivative of \eqref{eq:HH_VV_polar} with respect to $p_i$ and evaluating at zero will produce the minimum $\HH$ at position $p_i^*$ for agent $i$. The partial derivative of \eqref{eq:HH_VV_polar} with respect to $p_i$ is given by,
\[
\frac{\partial{\HH}}{\partial{p_i}} = M_{V_i}(p_i - C_{V_i})
\]
This demonstrates that the objective function $\HH$ in the 1-order case is minimal when $p_i$ is located at the centroid $C_{V_i}$.

For Voronoi partitions of the $k$-order, a similar approach to the order-$1$ partition can be followed. In \cite{JIANG201927} an objective function for higher-order coverage control with a general performance measure was introduced and a detailed derivation with performance measured defined by Euclidean distance for $k=2$ was presented. For completeness, the objective function is restated for arbitrary $k$. The objective function in terms of a $k$-order partition of $S$ is defined as,

%
%
%
%
%

\begin{equation} \label{eq:HH-RRk}
\HH (P,\RR) = \frac{1}{k}\sum_{I\in \II} \int_{R_I} f(q,p_1,\dots,p_k) \phi(q) dq.
\end{equation}
Where $f(q,\cdot)$ is the performance measure given by,
\[
f(q,p_1,\dots,p_k) = \sum_{i=1}^{k} \TwoNorm{q-p_i}^2.
\]
The objective function in terms of $k$-order Voronoi partitions is defined as,
%
%
%
\begin{align} \label{eq:HH_VVk}
\begin{split}
\HH (P) &= \frac{1}{k}\sum_{I\in \II} \int_{V_I} f(q,p_1,\dots,p_k) \phi(q) dq \\
&= \frac{1}{k}\sum_{I\in \II} \int_{V_I} \Big( \sum_{i=1}^{k} \TwoNorm{q-p_i}^2 \Big) \phi(q) dq.
\end{split}
\end{align}
Unlike like the first-order Voronoi objective function where the integration occurred over each cell and there was a cell for each agent, the $k$-order case does not have a one-to-one relationship between cells and agents. The performance measure is based on the distance $k$-agents are from each point $q$ in $V_I$. However, by manipulation, the objective function can be written in terms of the contribution of each agent separately. By distributing the integral,
\begin{align*}
\HH (P) = \frac{1}{k}\sum_{I\in \II} \Big[ & \int_{V_I}  \TwoNorm{q-p_{i_1}}^2  \phi(q) dq ~+ \dots \\ & \dots ~+ \int_{V_I} \TwoNorm{q-p_{i_k}}^2  \phi(q) dq \Big]
\end{align*}
and summing over all cells for each agent,
\begin{align*}
\HH (P) = \frac{1}{k}\sum_{I\in \II} & \int_{V_I}  \TwoNorm{q-p_{i_1}}^2  \phi(q) dq ~+ \dots \\ & \dots ~+ \frac{1}{k}\sum_{I\in \II}\int_{V_I} \TwoNorm{q-p_{i_k}}^2  \phi(q) dq,
\end{align*}
the higher-order objective function can be expressed in terms of the polar moment of inertia,
\begin{align} \label{eq:HH-VVk-polar}
\begin{split}
\HH (P) = \sum_{I\in \II} \Big[ J_{V,C_{V}} & + \frac{1}{k} M_{V_I}\TwoNorm{p_{i_1} - C_{V_I}}^2 + \dots \\ & + \frac{1}{k} M_{V_I}\TwoNorm{p_{i_k} - C_{V_I}}^2 \Big],
\end{split}
\end{align}
From (\ref{eq:HH-VVk-polar}), it can be seen that the value of $\HH$ depends on the distance from an agent to the centroid of a given cell. Clearly an agent cannot be located at the centroid of all the cells it is responsible for. To solve for the optimal location for agents to be located, the function $\HH$ is described in matrix form as follows,
\begin{align*} \label{eq:HH_VVk2_polar}
\HH (P) &= \onevec^\top (\mathbf{J}_{V_I,C_{V_I}}) \onevec \\
&+ \frac{1}{k}\big( p_{i_1}\onevec - \mathbf{C}_{V_{i_1}} \big)^\top \mathbf{M}_{V_{i_1}} \big( p_{i_1}\onevec - \mathbf{C}_{V_{i_1}} \big) \\
&+ \dots \\
&+ \frac{1}{k}\big( p_{i_k}\onevec - \mathbf{C}_{V_{i_k}} \big)^\top \mathbf{M}_{V_{i_k}} \big( p_{i_k}\onevec - \mathbf{C}_{V_{i_k}} \big).
\end{align*}
Where $\onevec$ is a vector of ones, $\mathbf{J}_{V_I,C_{V_I}}$ is a diagonal matrix, $\mathbf{C}_{V_{i}}$ is a vector of cell centroids associated with agent~$i$, and $\mathbf{M}_{V_{i}}$ is a diagonal matrix with elements on the diagonal represent the mass of the respective cell. Now the the optimal position $p_i^*$ for agent~$i$ can be solved by,
\begin{align*}
p_{i}^{*} &= \big(\mathbf{1}^\top \mathbf{M}_{V_i} \mathbf{1} \big)^{-1} \big( \mathbf{1}^\top \mathbf{M}_{V_i} \mathbf{C}_{V_i} \big) \\
&= \frac{\sum_{j=1}^{|\VV^i|} M_{V_j^i}\cm{V_j^i}}{\sum_{j=1}^{|\VV^i|} M_{V_j^i}} = C_{W_i}.
\end{align*}
$C_{W_i}$ is the centroid of the dominant cell $W_i$. As mentioned in the previous section, the cell $W_i$ is the dominant cell of agent $i$, which is the union of all the $k$-order Voronoi cells associated with agent $i$.
The objective function $\HH$ is minimal when $p_i$ is located at the centroid $C_{W_i}$ of the dominant cell $W_i$. This leads to the following lemma.
\begin{lemma} \label{lm:H-non-incr}
Given $P \in S^n$ and a k-order partition $\RR$ of $S$,
\[
\HH(P,\VV(P)) \leq \HH(P,\RR),
\]
i.e., the optimal partition is the $k$-order Voronoi partition. For $P' \in S$ with $\TwoNorm{p_{i}'-C_{W_i}} \leq \TwoNorm{p_i - C_{W_i}}$, $i \in \{1,\dots,n\}$,
\[
\HH(P',\RR) \leq \HH(P,\RR),
\]
i.e., the optimal positions of agents are the centroids.
\end{lemma}

As discussed in~\cite{JIANG201927}, for continuous control and communication the gradient descent control law is given by~$u_i = -k(p_i - C_{W_i})$ for gain~$k > 0$.
However, implementing this in continuous time assumes  that agents have exact position information about their neighbors at all times. Instead, we next discuss how to relax this requirement without resorting to a synchronous, periodic implementation.

\subsection{Communication between agents}\label{se:a2a-com}

We assume agent~$i$ has access to its own position~$p_i(t)$ at all times~$t \in \integernonnegative$, but must communicate with neighbors~$j \in \NN_i$ 
to obtain their positions~$p_j$. Similar to~\cite{nowzari2012self}, a request-response communication model is used where agent $i$ is able to request position~$p_j$ from agent~$j$ and agent~$j$ immediately responds with this information. We assume that packet loss does not occur and that round-trip latency is negligible such that agent~$i$ can request and receive information instantaneously i.e. the action of requesting and responding information occurs within the same timestamp.

More specifically, let $\{t_\ell^i\}_{\ell \in \integernonnegative} \subset \integernonnegative$ be the sequence of times at which agent~$i$ requests information from some neighbor~$j \in \NN_i$. Then, agent~$i$ only has access to the position of agent~$j$ at these times, e.g., at timestep~$t$ agent~$i$ has access to $\{p_j(t')\}_{t' \in \{t_\ell^i | t_\ell^i \leq t\}}$. 

\subsection{Agent state representation}\label{se:agent-data}
If an agent does not request information on every timestep then that agent does not have access to the current position of other agents. Therefore, agent~$i$ maintains state information pertaining to the most recent known position of agent~$j$ in addition to information that is able to model the evolution of uncertainty over time that exists with respect to agent~$j$'s current position. Given that agent $i$ has acquired position $p_j(t_\ell)$ from agent $j$ at timestep $t_\ell$, let $\tau > t_\ell$ be the amount of time that has elapsed since agent $i$ has communicated with agent $j$. Then the position $p_j(t)$ where $t = t_\ell + \tau$ will be unknown to agent $i$ at time $t$. However, if the maximum speed $\vmax$ for agent $j$ is known then agent $i$ can determine the set of all possible positions where agent $j$ could have traveled to in time duration $\tau$.  The set of possible positions for agent $j$ can be represented by a closed ball with center at $p_j(t_\ell)$ and radius $r_j = \vmax \tau_j$. To maintain state, each agent stores $p_j(t_\ell)$ and $r_j$ in memory for every agent in the network. The data storage for agent~$i$ is then defined by,
\begin{equation}\label{eq:agent-state-model}
\DD^i = \big((p_1^i,r_1^i),\dots,(p_n^i,r_n^i)\big)\in (S \times \realnonnegative)^n
\end{equation}
where $r_i^i=0$ for all time since it is assumed that agent~$i$ always has access to it's own position $p_i^i$ at every timestep. There exists two methods for which the contents of the data structure $\DD^i$ may be updated. The first is a time evolution update where all values $r_j^i$ increase in magnitude based on the the time duration $\tau$. The second update method, referred to as the information/position update, corresponds to the acquisition of a new position value $p_j^i$ via means of communication with agent~$j$. When a position update occurs for $p_j^i$, the value $r_j^i$ is reset i.e. $r_j^i=0$. This is due to the fact that the exact position of agent~$j$ is known at the instance in time that $p_j^i$ has been received and stored in memory by agent~$i$. In addition, two explicit methods for agent~$i$ to extract information from $\DD^i$. The first is the map $\map{\text{loc}}{(S \times \realnonnegative)^n}{S^n}$ that allows agent~$i$ to extract position information $(p_i^1, \dots, p_i^n))$ from $\DD^i$. The second extraction map $\map{\pi}{(S \times \realnonnegative)^n}{(S \times \realnonnegative)^m}$, where $m \leq n$, allows agent~$i$ to extract a subset $\pi(\DD^i) \subset \DD^i$ from data storage.

\subsection{Agent dynamics}\label{se:agent-dyn}
Considering the set $\ragents$ of agents moving in a convex polygon $\mathcal{S}$ with positions $P = (p_1,\dots,p_n)$. We consider discrete-time, single-integrator dynamics
\begin{align}\label{eq:dynamics}
p_i(t+1) = p_i(t) + u_i(t) \timestep ,
\end{align}
where~$\timestep > 0$ denotes the length of time of one timestep, and~$u_i(t)$ denotes the input at timestep~$t$ with $\TwoNorm{u_i(t)} \leq \vmax$ for each agent~$i \in \ragents$. The interest is in optimally deploying these agents in the domain~$S$ such that $k$ agents overlap responsibility for every point $q \in S$. Equipped with a communication model, a state data model, and agent dynamics the formal problem may now be presented by the following,
\begin{problem}\label{problem:main}
Given a set $\ragents = \{1,\dots, n\}$ of agents moving in a convex polygon $\mathcal{S} \subset \real^2$ with dynamics~\eqref{eq:dynamics}, maximum speed $\vmax > 0$, spatial density $\phi : \mathcal{S} \rightarrow \real$, and only depending on information local to agent $i$, find a distributed communication
and control strategy such that $p_i \rightarrow C_{W_i}$.
\end{problem}

Based on the data that each agent stores in memory, the exact computation of the $k$-order Voronoi cell cannot necessarily be achieved at each timestep. Next we address the issue of space partitioning with uncertainty for general cases of $k$-order.

\section{Space partition with uncertain information}\label{se:partition}
If agent~$i$ does not have access to the exact location $p_j$ of agent~$j$, then the uncertain position of agent~$j$ with respect to agent~$i$ can be represented to be within a set of points $D_j \in S$. This set $D_j$ represents all the possible points where agent~$j$ is guaranteed to be located relative to agent~$i$. The consequence of this representation is that agent~$i$ cannot compute it's dominant region exactly. However, because the position of agent~$j$ is guaranteed to be constrained to the set $D_j$, it is possible for agent~$i$ to compute regions in $S$ that pertain to a) the points that are certain to be part of its dominant cell, b) the points that are certain not to be part of agent~$i$'s dominant cell, and c) the region where it is uncertain if the points belong to agent~$i$'s dominant cell or not. The region of points that are certain to be part of agent~$i$'s dominant cell is referred to as the $k$-order guaranteed dominant cell of agent~$i$. The region of points that are certain to not be a part of agent~$i$'s dominant cell is referred to as the $k$-order dual-guaranteed dominant cell. Similar to the case of certain sites, we construct the guaranteed and dual-guaranteed dominant cell of an agent~$i$ by means of the $k$-order guaranteed and dual-guaranteed Voronoi cells. The the $k$-order guaranteed Voronoi partition is described next.
\subsection{k-order guaranteed Voronoi partitions}
\begin{figure*}
\centering%
\subfloat[$r=0$]{\label{fig:gV-ex-r0}\includegraphics[scale=0.16]{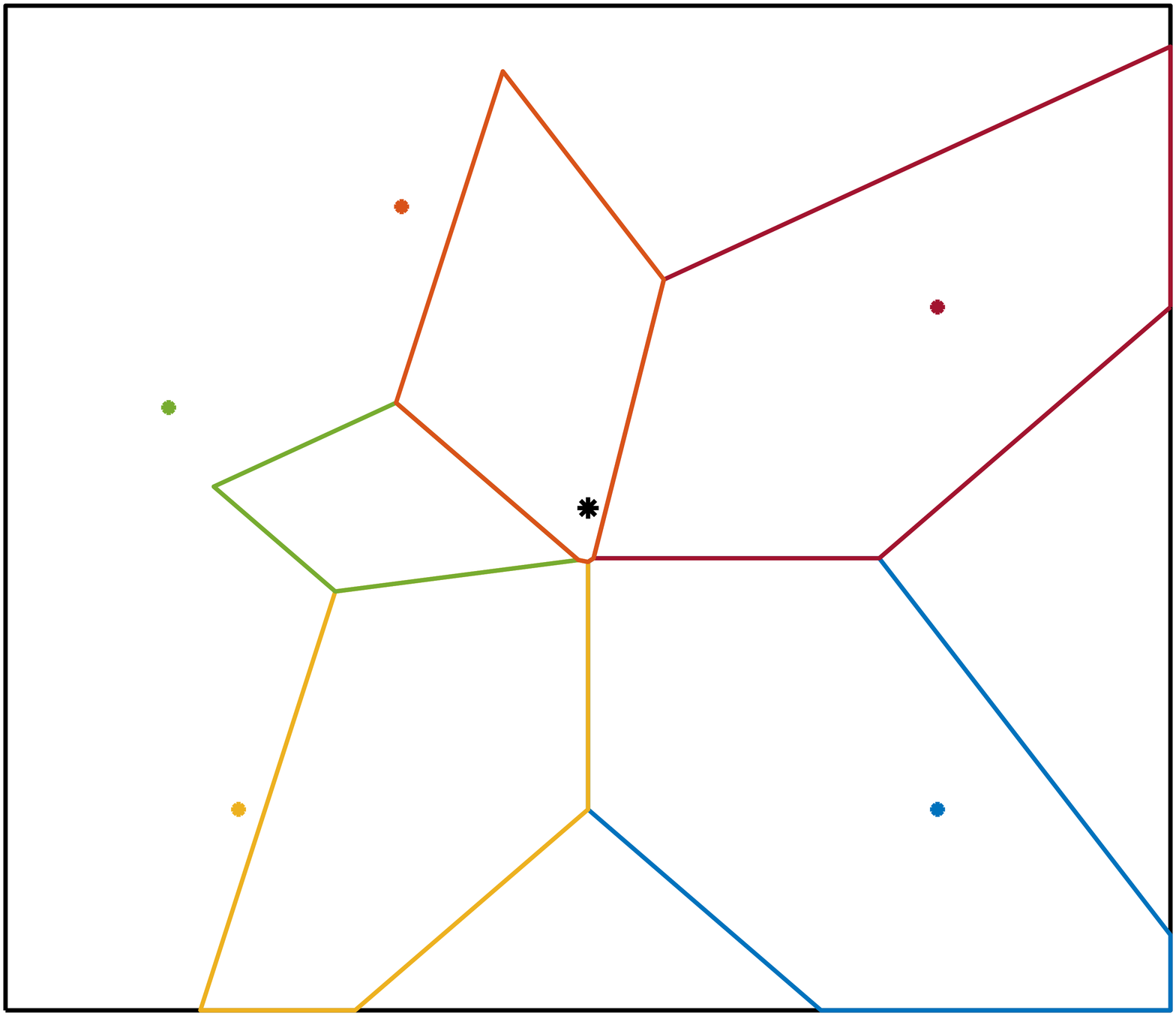}}%
\subfloat[$r=1$]{\label{fig:gV-ex-r1}\includegraphics[scale=0.16]{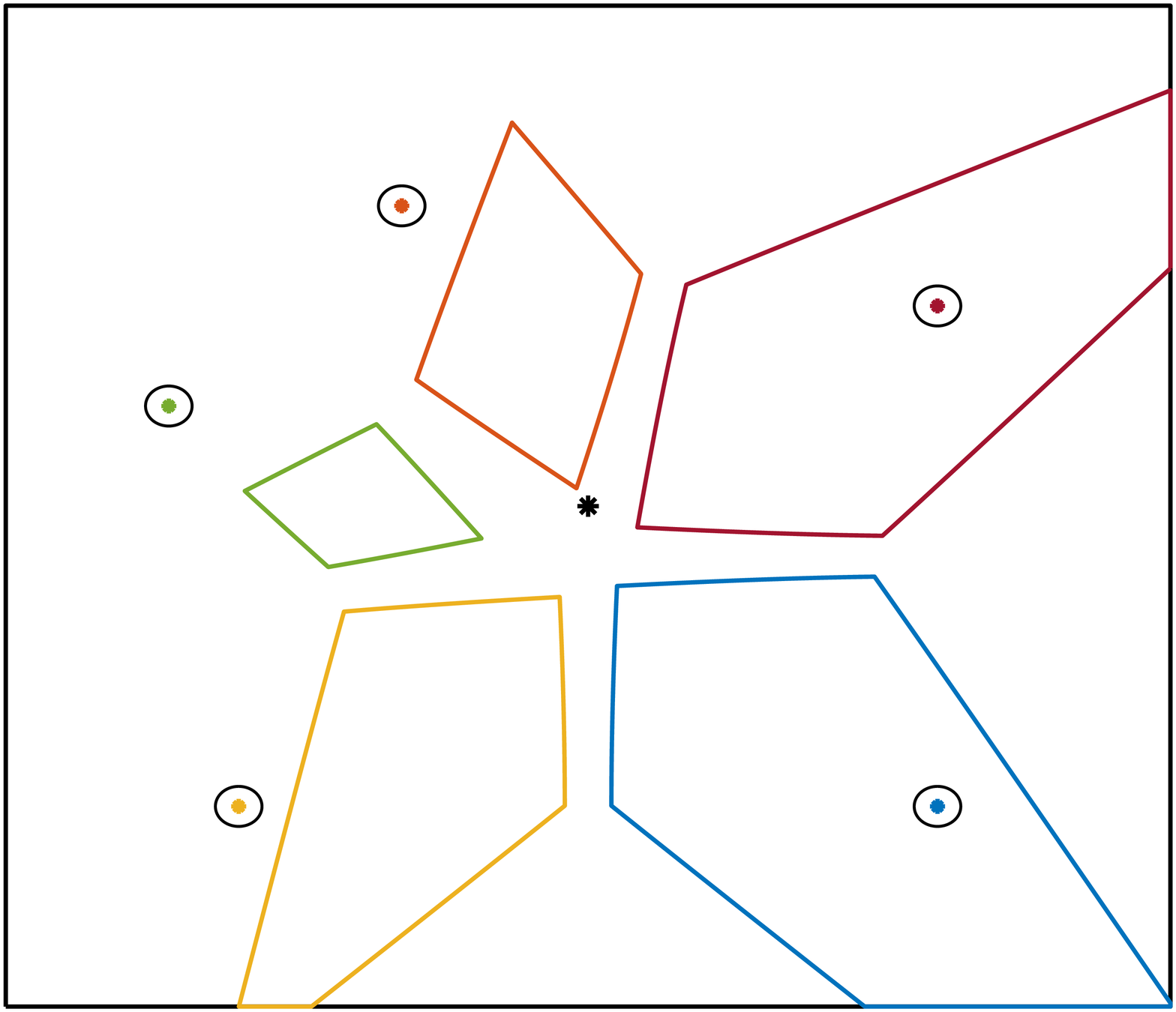}}%
\subfloat[$r=2$]{\label{fig:gV-ex-r2}\includegraphics[scale=0.16]{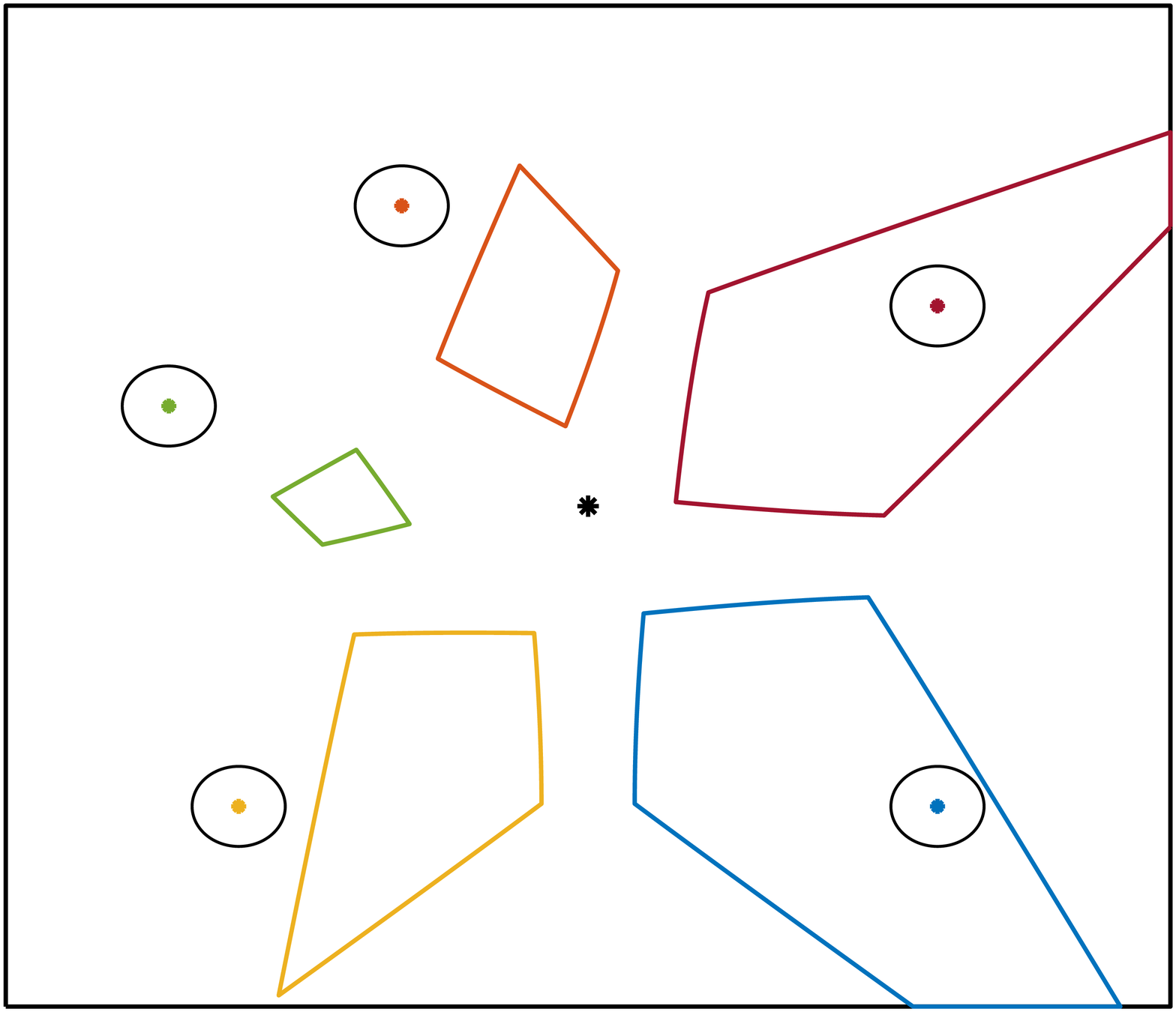}}%
\caption{The guaranteed $k$-order Voronoi cells ($k=2$) for a single agent represented by the black asterisk located close to the center of each diagram. Together, the diagrams illustrate the difference between cells when the radii changes, $r=0$~(a), $r=1$~(b), $r=2$~(c).}
\label{fig:gV-ex}%
\end{figure*}
To assist in the exposition that follows, the first-order guaranteed Voronoi cell is briefly mentioned. The first-order guaranteed Voronoi cell for agent~$i$ is given by,
\[
\gV_i = \Big\{ {q \in S}~\Big|~{\max_{x \in D_i}{\TwoNorm{q-x}} \leq \min_{y \in D_j}{\TwoNorm{q-y}}, ~\forall j \neq i} \Big\}.
\]
The cell $\gV_i$ contains the points of $S$ that are guaranteed to be closer to $p_i$ than to any other $p_j$, with $i \neq j$.
The uncertain regions $D_i$ and $D_j$ are considered to be closed balls $\cball{p_i,r_i}$ and $\cball{p_j,r_j}$ centered at $p_i$ and $p_j$ with radius $r_i$ and $r_j$, respectively. The set $D = \{D_1, \dots, D_n\}$ is the collection of uncertain regions for $n$~agents.
Similar to the discussion of $k$-order Voronoi partitions of certain sites where $I = (i_1,\dots,i_k) \in \II$ and $\PP_I \subset \PP$, a subset of $D$ is defined by $D_I = \{ D_1,\dots,D_k \}$
Given $D_I$ and with $D_J = D \setminus D_I$, the $k$-order guaranteed Voronoi cell associated with $I$ agents is defined by,
\begin{align*}
\gV_I = \Big\{ {q \in S}~\Big| 
~\max_{x \in D_i} & {\TwoNorm{q-x}} \leq \min_{y \in D_j}{\TwoNorm{q-y}} \\
& ~\forall D_i \in D_I, ~\forall D_j \in D_J  \Big\}
\end{align*}
The $k$-order guaranteed Voronoi cell represents the points that are guaranteed to be closer to the $k$-agents in $I$ with positions in $\PP_I$ than to the agents with positions in $\PP_J$. For example, with $k=2$ the $I$-agents becomes $I = \{D_i,D_j\}$ and has positions $\PP_I = \{p_i,p_j\}$ such that the second-order guaranteed Voronoi cell associated with agents $i$ and $j$ is given by,
\begin{align*}
\gV_{ij} = \bigg\{ {q \in S}~\Big|& ~{\max_{x \in D_i}{\TwoNorm{q-x}} \leq \min_{y \in D_J}{\TwoNorm{q-y}}},\\ & ~{\max_{x \in D_j}{\TwoNorm{q-x}} \leq \min_{y \in D_J}{\TwoNorm{q-y}}} \bigg\}
\end{align*}
with $D_J=D \setminus \{D_i,D_j\}$.
For agent~$i$, the guaranteed dominant region can be defined by,
\[
\gW_i = \bigcup_{I \in \II^i} \gV_I
\]
The cell $\gW_i$ represents the region that agent~$i$ is guaranteed to be responsible for covering.
\subsection{k-order dual-guaranteed Voronoi partitions}
\begin{figure*}
\centering%
\subfloat[$r=0$]{\label{fig:dgV-ex-r0}\includegraphics[scale=0.16]{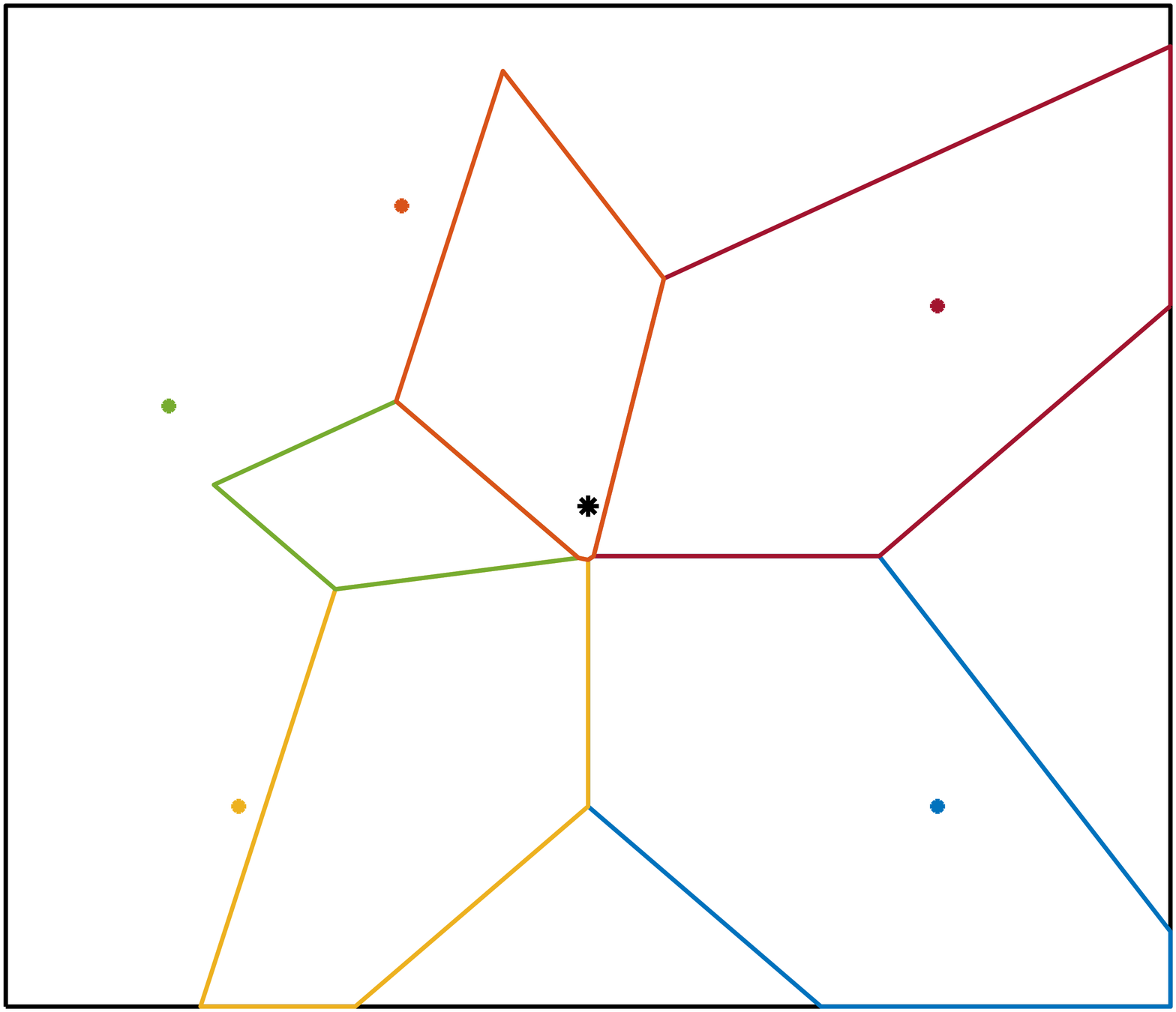}}%
\subfloat[$r=1$]{\label{fig:dgV-ex-r1}\includegraphics[scale=0.16]{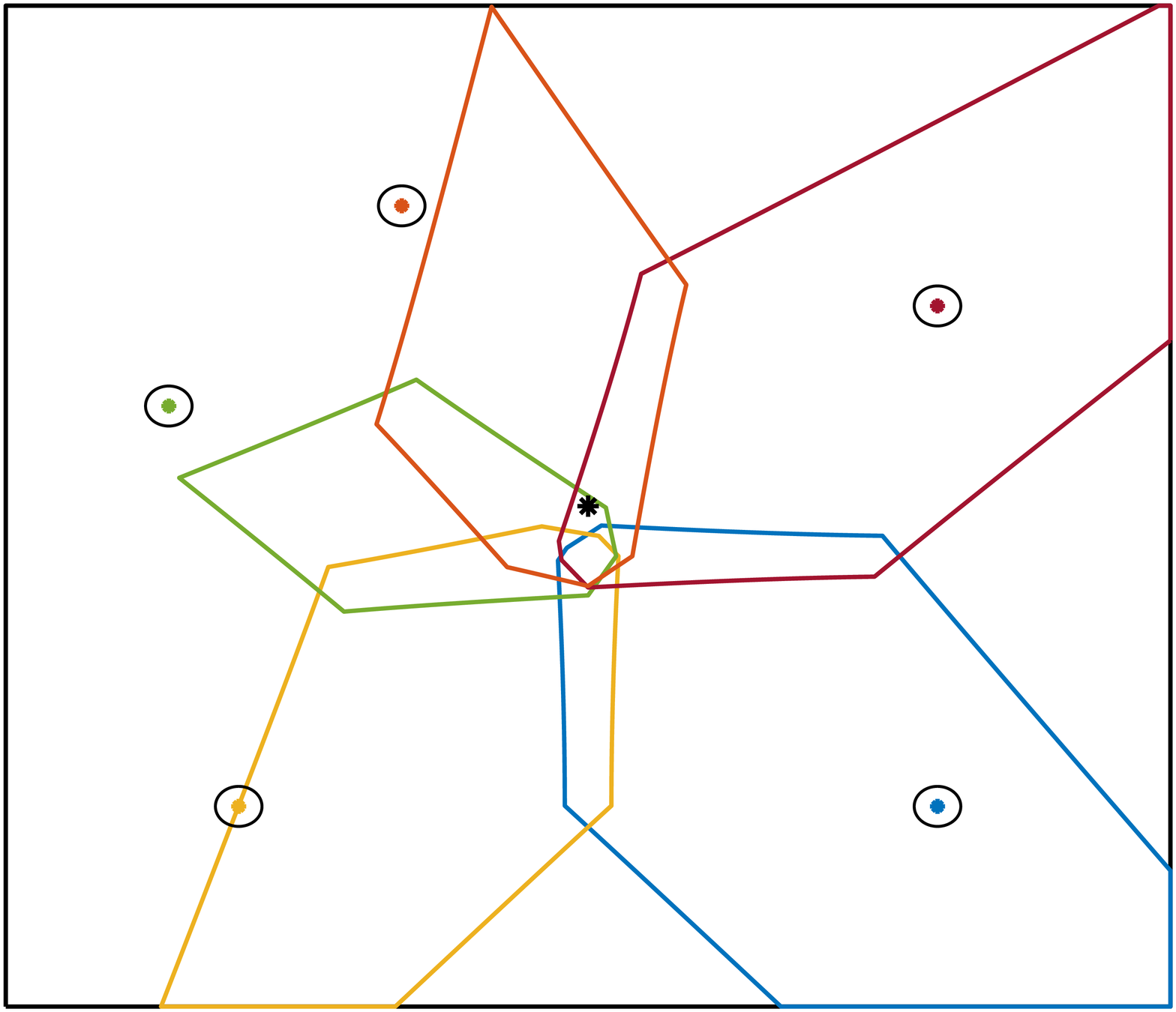}}%
\subfloat[$r=2$]{\label{fig:dgV-ex-r2}\includegraphics[scale=0.16]{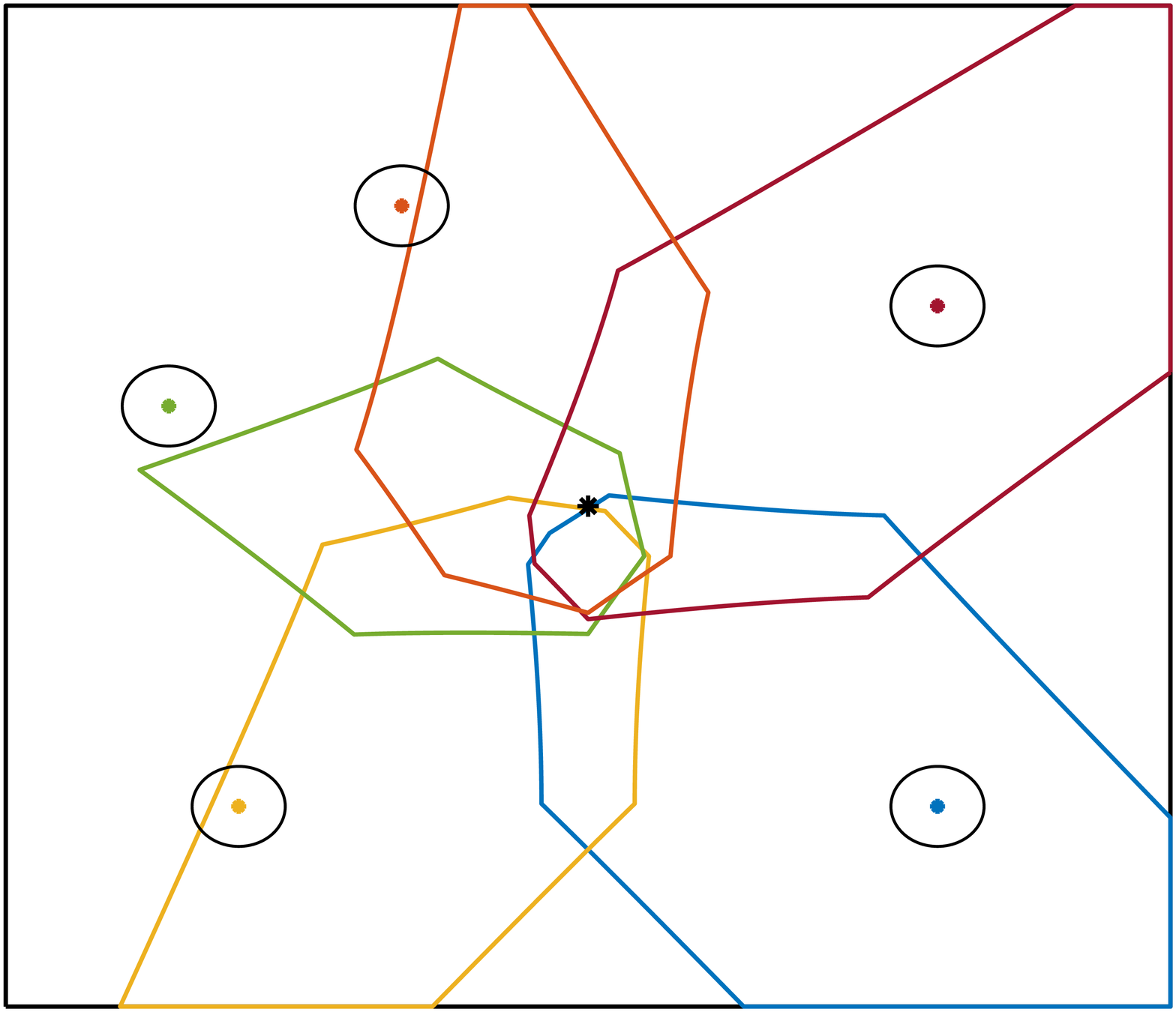}}%
\caption{The dual-guaranteed $k$-order Voronoi cells ($k=2$) for a single agent represented by the black asterisk located close to the center of each diagram. Together, the diagrams illustrate the difference between cells when the radii changes, $r=0$~(a), $r=1$~(b), $r=2$~(c).}
\label{fig:dgV-ex}%
\end{figure*}
In \cite{nowzari2012self}, the concept of dual-guaranteed Voronoi partitions was presented. Here we extend this concept to the case of $k$-order dual-guaranteed Voronoi partitions. Again using $D_I \subset D$ and $D_J = D \setminus D_I$, the $k$-order dual-guaranteed Voronoi cell for agents in $I$ is defined by,
\begin{align*}
\dgV_I = \Big\{ {q \in S}~\Big| 
~\min_{x \in D_i} & {\TwoNorm{q-x}} \leq \max_{y \in D_j}{\TwoNorm{q-y}}, \\
& ~\forall D_i \in D_I, ~\forall D_j \in D_J  \Big\}
\end{align*}
The region outside of the cell $\dgV_I$ represents the points that are guaranteed to be closer to agents in $J$ than to the agents in $I$. 
The dual-guaranteed dominant region associated with agent~$i$ is given by,
\[
\dgW_i = \bigcup_{I \in \II^i} \dgV_I
\]
The region outside of cell $\dgW_i$ represents the points that agent~$i$ is guaranteed not to be responsible for covering.

%


Next, a solution that includes both the design of a motion control law and a communication strategy for the above stated problem is presented.

\section{Self-triggered higher order coverage optimization}\label{se:design}
Given the problem described in Section~\ref{se:statement}, one possible approach would be for agent $i$ to periodically acquire position information from other agents. This would occur at each time step where agent $i$ would 1) acquire new information, 2) Compute it's dominant cell $W_i$, 3) compute the centroid $C_{W_i}$, 4) move towards $C_{W_i}$ at $\vmax$, 5) repeat. However, similar to continuous communication, a periodic method requires frequent communication and a potentially unnecessary computational burden. The method proposed in the following section  attempt to alleviate the communication and computational burden by a two part approach. The first component is a motion control law to determine how agents move when the information they possess is not up to date with respect to the most recent time step. The second component is an information update policy that allows each agent to decide when information from other agents should be acquired. 

\subsection{Motion control}\label{se:motion-control}
If agent~$i$ has access to the exact positions of other agents then agent~$i$ is capable of computing the exact dominant cell $W_i$. Consequently, agent~$i$ can compute the centroid $C_{W_i}$. Once $C_{W_i}$ has been computed, agent~$i$ may simply move towards it. When agent~$i$ does not communicate with the other agents in the network, the exact location of the other agents will be unknown to agent~$i$. Since the exact locations of other agents is unknown, each agent must rely on the data that it does possess as a means for deciding how to move. The data that an agent does possess at any given time step includes the most recent position update that it has received from the other agents and the time that has elapsed since the last update. 

Informally, the motion control law is described by the following. At each time step each agent uses the information that it has stored to compute it's $k$-order guaranteed and dual-guaranteed Voronoi cells. Next, each agent computes it's guaranteed and dual-guaranteed dominant cells. Once the agent has computed these cells, the agent then  computes the centroid for the guaranteed dominant cell $\gW_i$ and begins moving toward it.

The motion control law assumes that each agent has access to the value of the density $\phi$ over it's $k$-order guaranteed dominant cell. The motion control law as describe above does not necessarily guarantee that agents will move closer to the centroids of their dominant cells without applying additional constraints on agent movement. As in ~\cite{nowzari2012self}, the following lemma applies.
\begin{lemma}\label{lm:motionconstraint}
Given $p \neq q,q^* \in \real^2$, let $p' \in [p,q]$ such that $\TwoNorm{p' - q} \geq \TwoNorm{q^* - q}$, then $\TwoNorm{p' - q^*} \leq \TwoNorm{p - q^*}$.
\end{lemma}
Following lemma~\ref{lm:motionconstraint}, if $p = p_i$ is the position of agent $i$ that is moving toward $p'$ in the direction of the computed goal $q = C_{\gW_i}$ then the distance to $q^{*} = C_{W_i}$ decreases while
\begin{equation}\label{eq:motionconstraint}
\TwoNorm{p' - C_{\gW_i}} \geq \TwoNorm{C_{V_i} - C_{\gW_i}}
\end{equation}
holds. Since $C_{Wi}$ is unknown to agent $i$ the right hand side of ~\ref{eq:motionconstraint} cannot be computed. However the value $\TwoNorm{C_{W_i} - C_{\gW_i}}$ can be bounded such that
\begin{equation} \label{eq:algconstraint}
\TwoNorm{p' - C_{\gW_i}} \geq \text{bnd}_i
\end{equation}
where $\text{bnd}_i$ is given by
\begin{equation} \label{eq:algbound}
\text{bnd}_i = \text{bnd}(\gW_i,\dgW_i) = 2 cr_{\dgW_i} \bigg( 1 - \frac{M_{\gW_i}}{M_{\dgW_i}} \bigg)
\end{equation}
Therefore, agent $i$ moves towards $C_{\gW_i}$ as much as possible in one timestep while maintaining the condition in (\ref{eq:algconstraint}). The motion control law is formally defined in table (\ref{tab:motion}). 

For every consecutive time step that an agent goes without receiving updated information, (\ref{eq:algbound}) increases making the condition of (\ref{eq:algconstraint}) less likely to be achievable. Leading to the condition where the agent can no longer move in a manner that does not increase the distance to $C_{W_i}$. Therefore, a decision mechanism that governs when an agents will acquire new information is required and is discussed next.

\begin{algorithm}[ht]
  {\footnotesize 
	  Agent $i \in \until{n}$ performs:
      \begin{algorithmic}[1]
      \STATE set $D = \mathcal{D}^i$
      \STATE compute $L= \gW_i(D)$ and $U=\dgW_i(D)$
      \STATE compute $q=C_L$ and $r=\text{bnd}(L,U)$
      \STATE set $d = \vmax \timestep$
      \STATE set $p_i' = \tbb(p_i,d, q, r)$
      \STATE move to $p_i'$
      \STATE set $r_j^i = r_j^i + d$
      \STATE set $\mathcal{D}_j^i = (~p_j^i, \min{\{r_j^i, \diam(S)}\}~)$
      \STATE set $\mathcal{D}_i^i = (p_i',0)$
      \end{algorithmic}
  }
  \caption{\hspace*{-.5ex}: \small \algoMotionControl}\label{tab:motion}
\end{algorithm}

\subsection{Update decision policy}\label{se:decision-policy}
The second major aspect of the self-triggered deployment strategy provides a decision mechanism that determines when an agent must perform an information update via communication with other agents. Updates to position information will be necessary for an agent to reduce the level of uncertainty that it has accumulated since the last time an update occurred. As previously mentioned, as time elapses without receiving position information from other agents, the true location of $C_{W_i}$ will be unknown and the set of possible locations for $C_{W_i}$ will continue to increase in size. Based on the motion control law presented in the previous section, agent~$i$ will rely on moving towards $C_{\gW_i}$ so long as condition (\ref{eq:algconstraint}) holds. If it becomes infeasible for agent~$i$ to move due to condition (\ref{eq:algconstraint}) not being satisfied, then agent~$i$ must perform an information update at that moment in time in order to maintain condition (\ref{eq:algconstraint}). Therefore, the update decision policy can be describe as follows. For every timestep, each agent computes their $k$-order guaranteed and dual guaranteed dominant cells, as well as computing the bound (\ref{eq:algbound}). Then each agent decides whether or not to perform a position data  update. Agent~$i$ will decide to perform the update when the bound (\ref{eq:algbound}) becomes greater than or equal to $\TwoNorm{p_i - C_{\gW_i}}$. It is possible that the points $p_i$ and $C_{\gW_i}$ may become close to one another i.e. $\TwoNorm{p_i - C_{\gW_i}} < \epsilon$ for $\epsilon > 0$.In this case, the bound (\ref{eq:algbound}) may not be able to become small enough such that a position update is not required. To handle this condition, the value of $\TwoNorm{p_i - C_{\gW_i}}$ is clamped at $\epsilon$ so that a minimum amount of time will pass before and update will occur. The update policy is described formally in table~\ref{tab:one-update}. 

\begin{algorithm}[htb]
{\footnotesize
  Agent $i \in \until{n}$ performs:
  \begin{algorithmic}[1]
  	\STATE set $D = \mathcal{D}^i$
    \STATE compute $L= \gW_i(D)$ and $U=\dgW_i(D)$
    \STATE compute $q=C_L$ and $r=\text{bnd}(L,U)$
    \IF {$r \geq \max{\{\TwoNorm{q-p_i},\epsilon\}}$} 
    \STATE reset $\mathcal{D}^i$ by performing a position update
    \ENDIF
  \end{algorithmic}}
  \caption{\hspace*{-.5ex}: \small \algoOneStep}\label{tab:one-update}
\end{algorithm}

\subsection{The ~$k$-order self-triggered centroid algorithm}
A self-triggered deployment strategy can be formulated by combining the motion control law defined in Table \ref{tab:motion} and the update decision policy from Table \ref{tab:one-update}. First, it is noted that combining the two algorithms from Table \ref{tab:motion} and Table \ref{tab:one-update} without modification would provide an event-triggered deployment strategy. The event-triggered strategy would be performed on each timestep where agent~$i$ runs the update decision policy followed by running the motion control law. This requires agent~$i$ to compute $L$, $U$, $C_L$, and $r$ from Table \ref{tab:one-update} on every timestep.
However, agent~$i$ is in possession of all the information necessary to predict its motion trajectory up to the time in the future where $r \geq \max{\{\TwoNorm{q-p_i},\epsilon\}}$ occurs. The self-triggered algorithm is presented in Table \ref{tab:n-step}. In addition, note that a trivial update mechanism would provide each agent with up-to-date locations for all other agents in the network i.e. using all information stored in $\DD^i$. However, this is costly from a communications point of view. Instead, a localized algorithm is proposed that limits the number of agents that agent~$i$ must acquire information from. To compute $\gW_i$ and $\dgW_i$, agent~$i$ must have knowledge of only a subset of agent positions. The subset of agents used by agent $i$ can be found by first defining
\[
\ragents^i(q) = \setdef{j \in \ragents}{\TwoNorm{p_j-q} < \TwoNorm{p_i-q},j \neq i}
\]
where $|\ragents^i(q)| \geq k$. 
Based on this definition we can redefine the cell $W_i$ by
\[
W_i = \setdef{q \in S}{(|\ragents^i(q)|) \leq k-1}
\]
 To locally compute $W_i$ at the specific time when step \algostep{4} is executed, the \algoVoronoi is used. This is borrowed from \cite{LuoLXWH12icdcs} and presented in Algorithm~\ref{tab:voronoi}
 
\begin{algorithm}[ht]
  {\footnotesize
    \begin{algorithmic}[1]
      \STATE initialize $\rho = 0$
      \REPEAT 
        \STATE {set $out \gets true$}
      	\STATE {set $\rho \gets \rho + \gamma$}
      	\STATE {set $\NN_i(\rho) \gets \setdef{j}{\TwoNorm{p_j-p_i} < \rho}$}
      	\FORALL {$\setdef{q \in S}{\TwoNorm{q-p_i} = \rho/2}$} 
      	  \STATE {set $\ragents^i(q) \gets \setdef{j \in \NN_i(\rho)}{\TwoNorm{p_j-q} < \TwoNorm{p_i-q},j \neq i}$}
      	  \IF {$|\ragents^i(q)| < k$} 
		    \STATE {set $out \gets false$}
			\STATE {break} 
		  \ENDIF 
        \ENDFOR
      \UNTIL{$out = true$}
      \STATE {compute $W_i$ from $\NN_i(\rho)$}
    \end{algorithmic}}
  \caption{\hspace*{-.5ex}: \small \algoVoronoi}\label{tab:voronoi} 
\end{algorithm}

The \algoVoronoi is based on agent $i$ gradually increasing its communication radius until all the information required to  construct its exact $k$-order Voronoi cell has been obtained. Combining Algorithms~\ref{tab:motion}-\ref{tab:voronoi} leads to the complete \algoFull described in Algorithm~\ref{tab:full}.
\begin{algorithm}[htb]
{\footnotesize
  Agent $i \in \until{n}$ performs:
  \begin{algorithmic}[1]
  	\STATE set $D = \mathcal{D}^i$
    \STATE compute $L= \gW_i(D)$ and $U=\dgW_i(D)$
    \STATE compute $q=C_L$ and $r=\text{bnd}(L,U)$
    \IF {$r \geq \max{\{\TwoNorm{q-p_i},\epsilon\}}$} 
    	\STATE reset $\mathcal{D}^i$ by performing a position update
    \ELSE
    \STATE initialize $t_{sleep} = 0$
    \WHILE {$r < \max{\{\TwoNorm{q-p_i},\epsilon\}}$}
    	\STATE set $t_{sleep} = t_{sleep} + 1$
      	\STATE set $d = \vmax \timestep$
      	\STATE set $p_i' = \tbb(p_i,d, q, r)$
      	\STATE move to $p_i'$
      	\STATE set $r_j^i = r_j^i + d$
      	\STATE set $\mathcal{D}_j^i = (~p_j^i, \min{\{r_j^i, \diam(S)}\}~)$
      	\STATE set $\mathcal{D}_i^i = (p_i',0)$
      	\STATE set $D = \mathcal{D}^i$
    	\STATE compute $L= \gW_i(D)$ and $U=\dgW_i(D)$
    	\STATE compute $q=C_L$ and $r=\text{bnd}(L,U)$
    \ENDWHILE
    \STATE wait for $t_{sleep}$ timesteps
    \STATE repeat
    \ENDIF
  \end{algorithmic}}
  \caption{\hspace*{-.5ex}: \small \algoMultiStep}\label{tab:n-step}
\end{algorithm}

\begin{algorithm}[ht]
{\footnotesize
  \begin{flushleft}
    Initialization
  \end{flushleft}
  \begin{algorithmic}[1]
    \STATE execute \algoVoronoi
  \end{algorithmic}
  \begin{flushleft}
    At time step $\ell \in \integernonnegative$, agent $i \in \until{n}$ performs:  
  \end{flushleft}
  \begin{algorithmic}[1]
  	\STATE set $D = \pi({\DD}^i)$
    \STATE compute $L= \gW_i(D)$ and $U=\dgW_i(D)$
    \STATE compute $q=C_L$ and $r=\text{bnd}(L,U)$
    \IF {$r \geq \max{\{\TwoNorm{q-p_i},\epsilon\}}$} 
    	\STATE reset $\mathcal{D}^i$ by running \algoVoronoi
    	\STATE set $D = \pi ({D}^i)$
    	\STATE compute $L= \gW_i(D)$ and $U=\dgW_i(D)$
    	\STATE compute $q=C_L$ and $r=\text{bnd}(L,U)$
    \ENDIF
    \STATE set $d = \vmax \timestep$
    \STATE set $p_i' = \tbb(p_i,d, q, r)$
    \STATE move to $p_i'$
    \STATE set $r_j^i = r_j^i + d$
    \STATE set $\mathcal{D}_j^i = (~p_j^i, \min{\{r_j^i, \diam(S)}\}~)$
    \STATE set $\mathcal{D}_i^i = (p_i',0)$
  \end{algorithmic}}
\caption{\hspace*{-.5ex}: \small \algoFull}\label{tab:full}
\end{algorithm}

\begin{figure*}
\centering%
\subfloat[Initial configuration]{\includegraphics[scale=0.16]{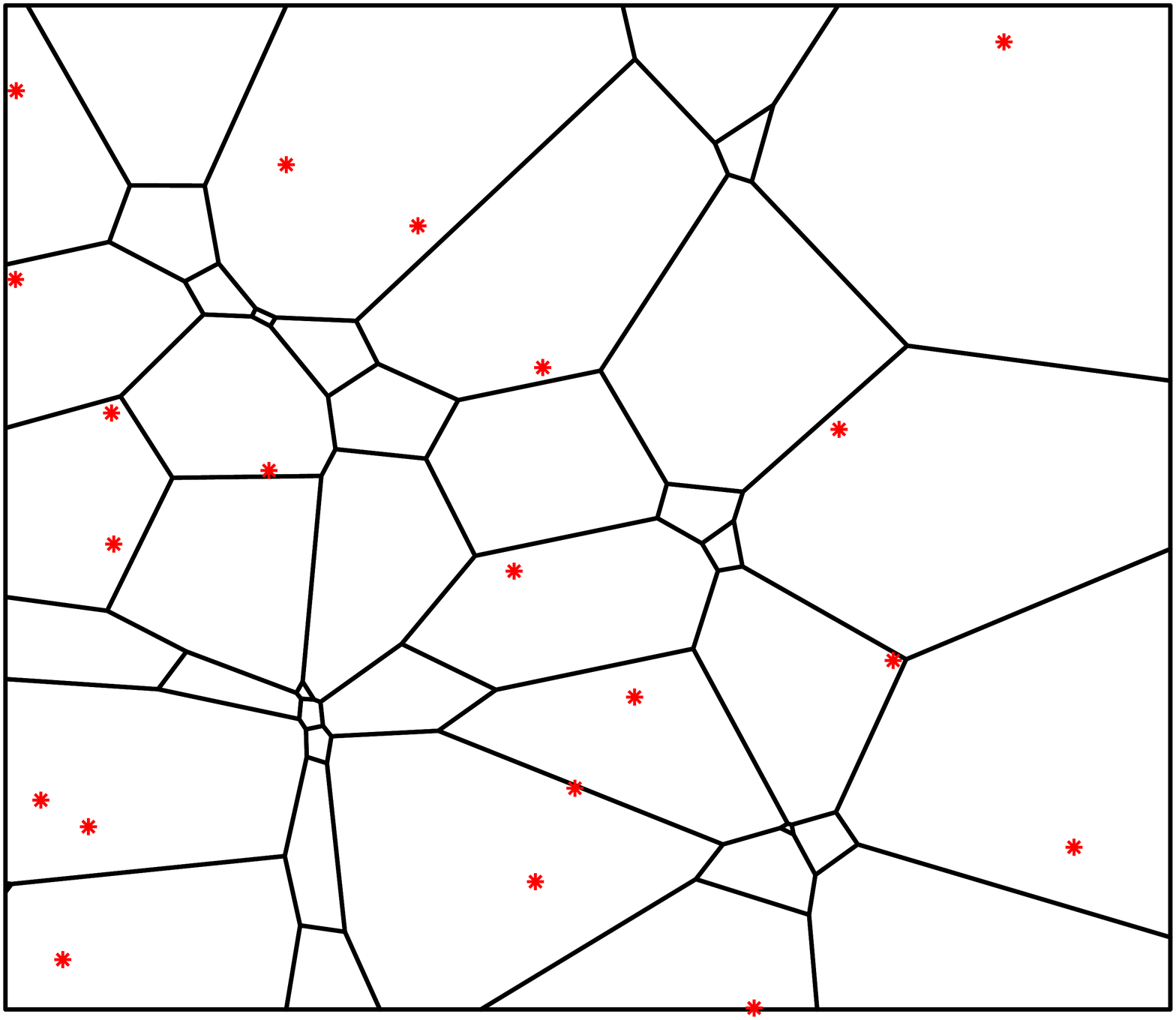} \label{fig:initial-20}}%
\subfloat[Trajectories]{\includegraphics[scale=0.16]{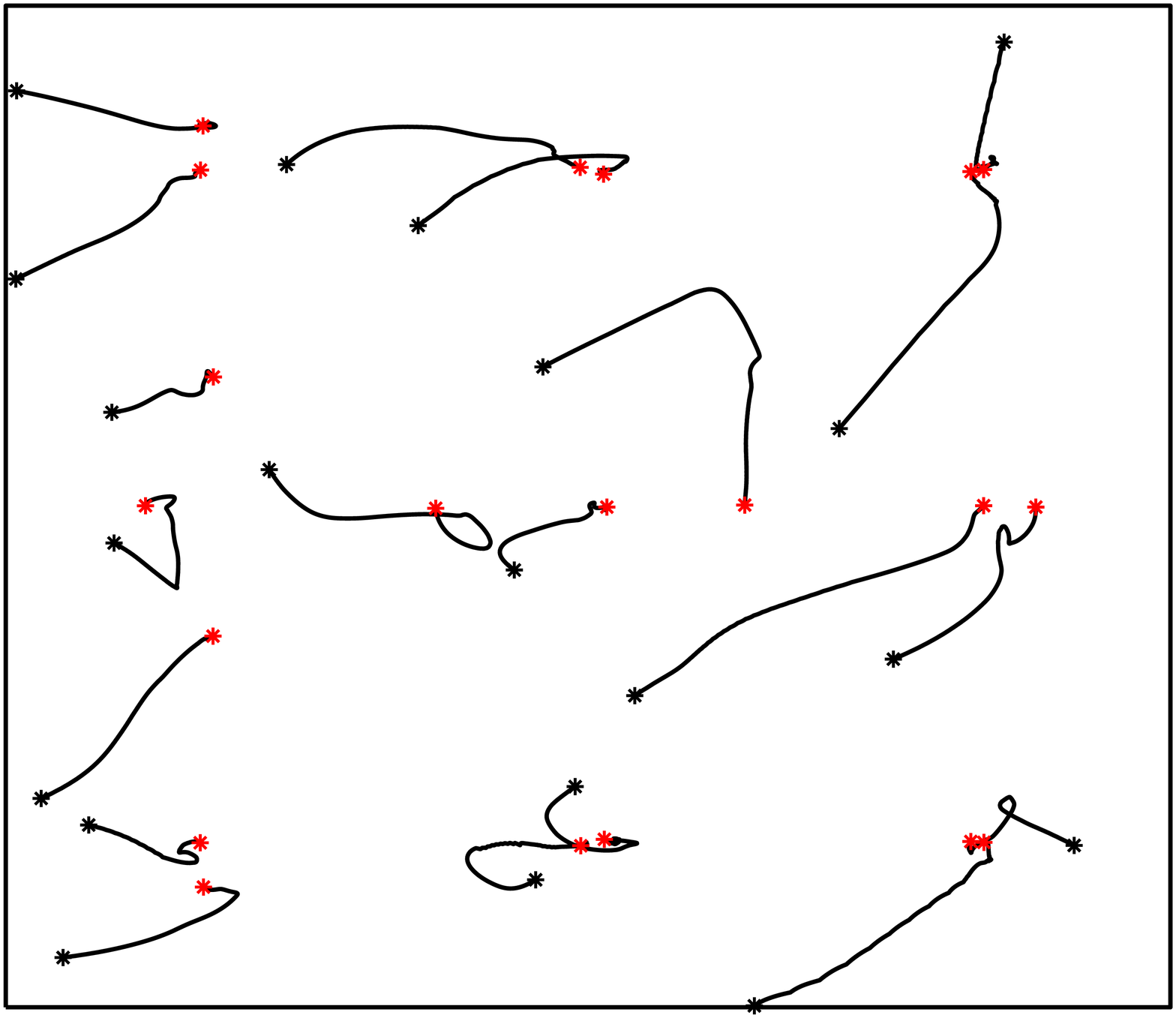} \label{fig:traj-20}}%
\subfloat[Final configuration]{\includegraphics[scale=0.16]{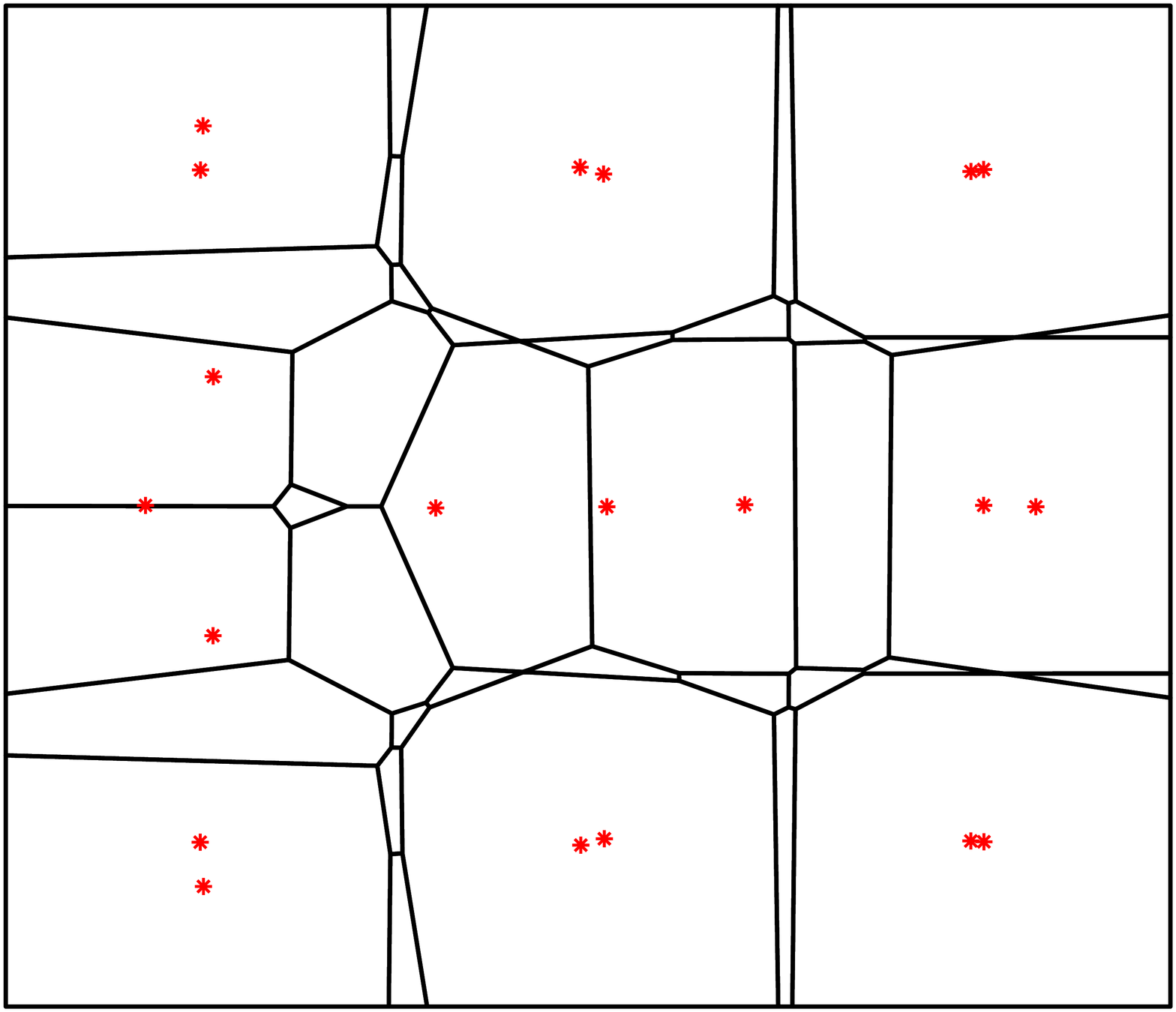} \label{fig:final-20}}%
\caption{Initial configuration (a), trajectories (b), and final configuration (c) for 20 agents running the \algoFull}%
\label{fig:sim-20}%
\end{figure*}
\section{Convergence analysis}\label{se:convergence}
A detailed analysis is provided in this section to demonstrate that agents following the motion and information update strategies presented thus far will generate a network configuration such that all agents converge to their centroidal positions. The asynchronous timing of information exchange that occurs during the network evolution is dependent on the the number of agents in the network, the area of the task space, and the initial agent configuration. This presents challenges when attempting to analyze the convergent properties in a similar fashion to that of the continuous-time continuous-update policy. Instead, our analysis assumes that agents move according to the motion control law given in Table (\ref{tab:motion}) while considering information updates that occur randomly in time. We show that regardless of how agents share information, trajectories governed by the motion control law, in particular the constraints laid out by $\tbb$, will at least converge to a positively invariant set and that if the update decision policy is followed, the network will converge to the centroidal configuration. To achieve this, a set-valued map $T$ is defined that describes the evolution of the network state represented by the data storage of all agents. Then by applying the LaSalle Invariance Principal for set-valued maps, it is shown that all trajectories generated under the state evolution map $T$ provide values of the performance function $\HH$ that are monotonically non-increasing. It is also shown that there exist under $T$ a weakly positively invariant set that is specifically contained in the trajectories that follow the update decision policy of Table (\ref{tab:one-update}). Finally, we deduced that this set coincides with the centroidal network configuration in task space. This is proposed formally by the following:
\begin{proposition} \label{prop:main}
For $\epsilon \in [0 \diam(S)]$, the agent position evolving under the self-triggered deployment algorithm from any initial network configuration in $S$ converges to the $k$-order Voronoi centroidal configuration.
\end{proposition}
To proceed, we formally define $\DD = (\DD^1,\dots,\DD^n) \in (S \times \realnonnegative)^{n^2}$ as the sate of an $n$~agent network where $\DD^i = \big((p_1^i,r_1^i),\dots,(p_n^i,r_n^i)\big)$ is the state of agent~$i$. We define $\map{\MM}{(S \times \realnonnegative)^{n^2}}{(S \times \realnonnegative)^{n^2}}$ as the map that updates both the motion $p_i^i$ and uncertainty evolution $r_j^i$ in $\DD$. Recall that the magnitude of $r_j^i$ increases over time when information updates do not occur. We define $\map{f_u}{(S \times \realnonnegative)^{n^2}}{(S \times \realnonnegative)^{n^2}}$ as a mapping of the network state into itself and it describes the information update evolution when the update decision policy from Table (\ref{tab:one-update}) is followed. Note that the self-triggered algorithm can be described as the composition $f_{st} = \MM \circ f_u = \MM(f_u(\DD))$. Let $\svmap{\UU}{(S \times \realnonnegative)^{n^2}}{(S \times \realnonnegative)^{n^2}}$ be the set-valued map that represents any possible information-update evolution. For $\DD'\in \UU(\DD)$, the $i$th component of $\DD'$ is described by,
\[
  \DD^i = \begin{cases}
           \big((p_1^i,r_1^i),\dots,(p_n^i,r_n^i)\big), & \text{no update} \\
           \big((p_1^i,0),\dots,(p_n^i,0)\big), & \text{update occurred} 
         \end{cases}
\]
Note that $f_u(\DD) \in (S \times \realnonnegative)^{n^2}$ is an element of the domain, but $\UU(\DD) \subset (S \times \realnonnegative)^{n^2}$ is a subset of the domain and further, $f_u(\DD) \in \UU(\DD)$ is one outcome in $\UU(\DD)$.

Given the definition of $\MM$ and $\UU$, the full state evolution  is defined by the set-valued map $\svmap{T}{(S \times \realnonnegative)^{n^2}}{(S \times \realnonnegative)^{n^2}}$ where $T = \UU \circ \MM$. Since $\UU$ is closed and $\MM$ is continuous, the evolution map $T$ is closed. For a trajectory $\gamma = \{\DD(t_\ell)\}_{t \in \integerpositive}$ generated by the self-triggered algorithm and $\gamma' = \{\DD'(t_\ell)\}_{t \in \integerpositive}$ given by $\DD'(t_\ell) = f_u(\DD(t_\ell))$ then,
\begin{equation} \label{eq:T-traj}
\DD'(t_{\ell + 1}) = T(\DD'(t_\ell))
\end{equation}

Let $\map{\text{loc}}{(S \times \realnonnegative)^n}{S^n}$ be a map that extracts the positions $P=(p_1^1, \dots, p_n^n))$ from $\DD$ such that $\HH(\loc{\DD}) = \HH(P)$.

\begin{lemma} \label{lm:mono-non-inc}
$\map{\HH}{(S \times \realnonnegative)^{n^2}}{\real}$ is monotonically non-increasing along the trajectories of $T$.
\end{lemma}
\begin{proof}
Let $\DD \in \dataspace$ and $\DD' \in T(\DD)$. Let $P =\loc{\DD}$ and $P'=\loc{\DD'}=\loc{\MM(\DD))}$. To demonstrate that $\HH(P') \leq H(P)$, first the $k$-order partition $\VV(P)$ is fixed. Then for each $i \in \ragents$, if  the condition $\TwoNorm{p'_i-C_{\gW_i}} \leq bnd(\DD^i)$ is true then $p'_i = p_i$. This is due to the fact that agent $i$ strictly follows the definition of $\tbb$. If instead $\TwoNorm{p'_i-C_{\gW_i}} > bnd(\DD^i)$ then it is true that $\TwoNorm{p'_i - C_{W_i}} < \TwoNorm{p_i - C_{W_i}}$ by lemma \ref{lm:motionconstraint} and (\ref{eq:algbound}). For both cases, $\HH(P',\VV(p)) \leq \HH(P,\VV(P))$ and furthermore, from lemma \ref{lm:H-non-incr}, $\HH(P',\VV(P')) \leq \HH(P',\VV(P))$
\end{proof}

\begin{lemma} \label{lm:weak-pos-invar}
Let $\gamma'$ be a trajectory of (\ref{eq:T-traj}). Then the $\omega$-limit set $\Omega(\gamma') \subset \dataspace$ with $\Omega(\gamma') \neq \emptyset$ belongs to $\HH^{-1}(c)$ for some constant $c \in \realnonnegative$ and is weakly positively invariant. Let $\gamma'$ be a trajectory of (\ref{eq:T-traj}).
\end{lemma}
\begin{proof}
Let $\gamma'$ be a trajectory of (\ref{eq:T-traj}). First, note that $\gamma'$ being bounded implies $\Omega(\gamma') \neq \emptyset$ and for $\DD' \in \Omega(\gamma')$ there exists a converging sub-sequence $\setdef{\DD'(t_{\ell_m}}{m \in \integernonnegative}$ of $\gamma'$ such that $\DD'(t_{\ell_m}) \rightarrow \DD'$ as $m \rightarrow \infty$. In addition, the sequence $\setdef{\DD'(t_{\ell_{m}+1})}{m \in \integernonnegative}$ is also bounded and has a converging sub-sequence where for $\hat{\DD}'$ the sequence $\DD'(t_{\ell_{m+1}} \rightarrow \hat{\DD}'$ for $m \rightarrow \infty$. Since by definition $\hat{\DD}' \in \Omega(\gamma')$ and $T$ is closed, this implies $\Omega(\gamma')$ is weakly positive invariant. Since $\gamma$ is bounded and $\HH$ is non-increasing along $\gamma$ for all of $\dataspace$, the sequence $\HH \circ \gamma = \setdef{\HH(\gamma(l))}{l \in \integernonnegative}$ is decreasing and bounded from below and therefore convergent. Since for any $z \in \Omega(\gamma)$ there is a converging subsequence $\gamma(\ell_m)$ in $\Omega(\gamma)$ that converges to $z$ and since $\HH$ is continuous, $\HH(\gamma(\ell_m)) \rightarrow \HH(z)=c$ as $m \rightarrow \infty$ where $c \in \real$ is a constant.
\end{proof}
\begin{figure*}
\centering%
\subfloat[Initial configuration]{\includegraphics[scale=0.16]{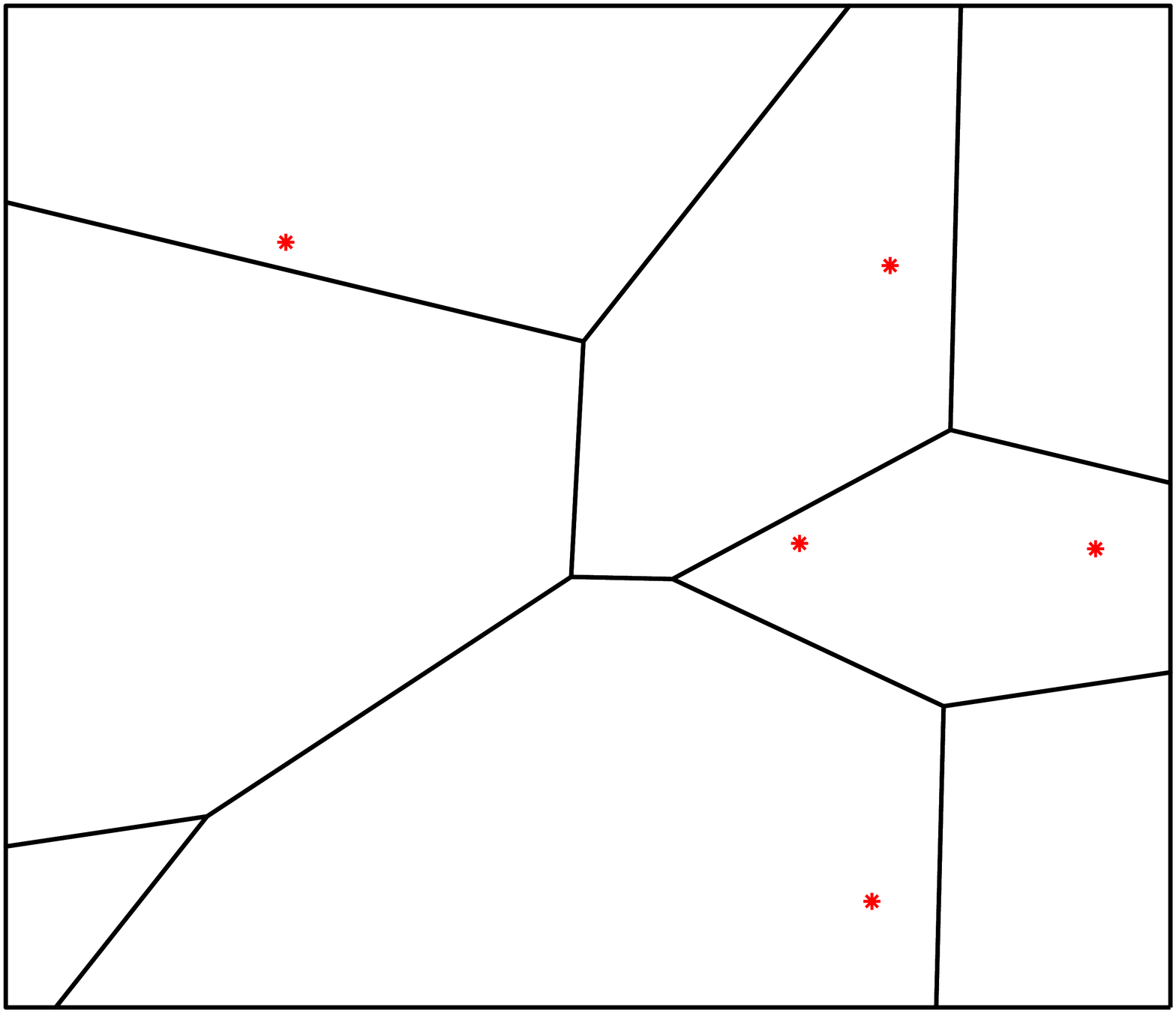} \label{fig:initial-5}}%
\subfloat[Trajectories]{\includegraphics[scale=0.16]{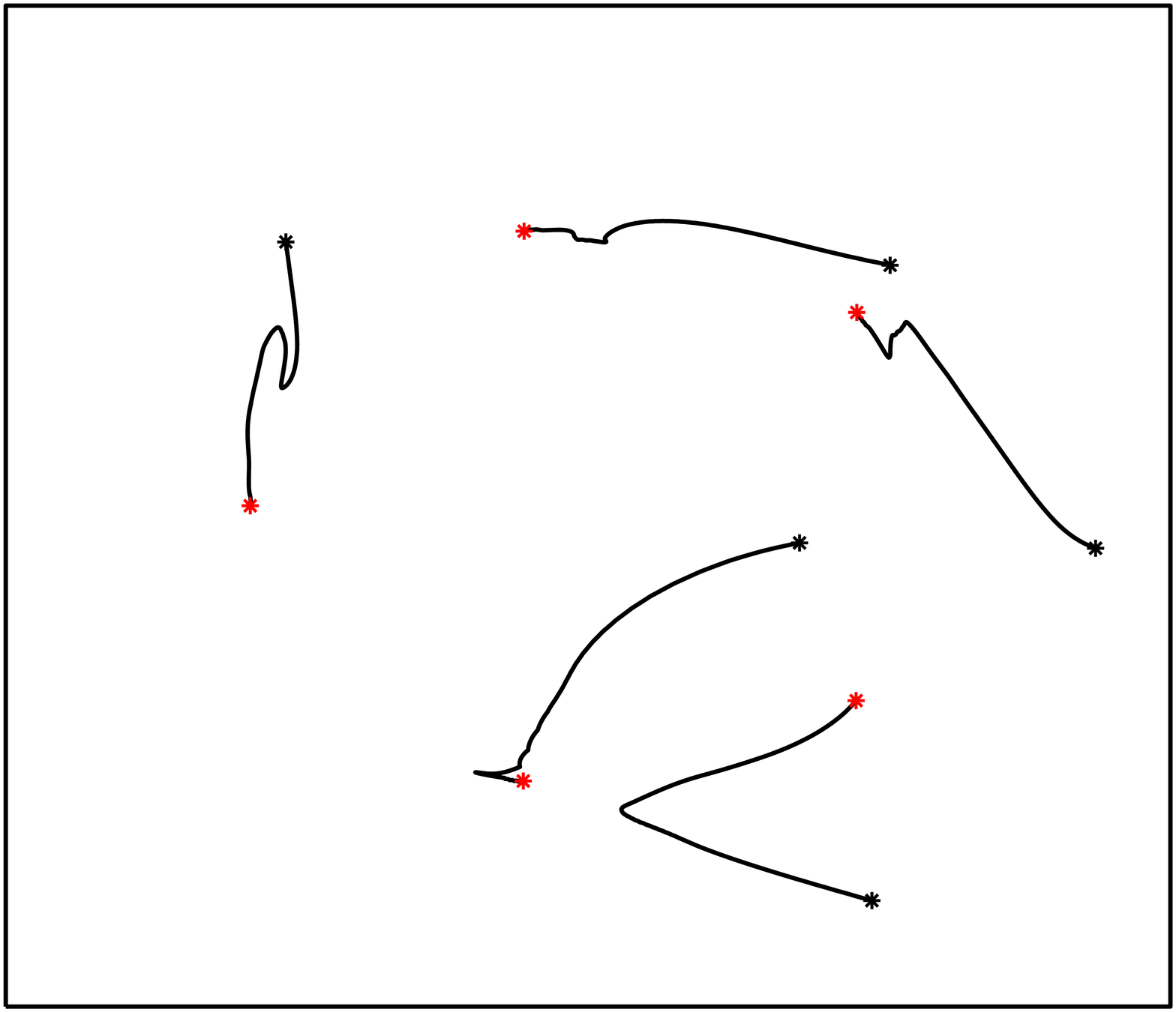} \label{fig:traj-5}}%
\subfloat[Final configuration]{\includegraphics[scale=0.16]{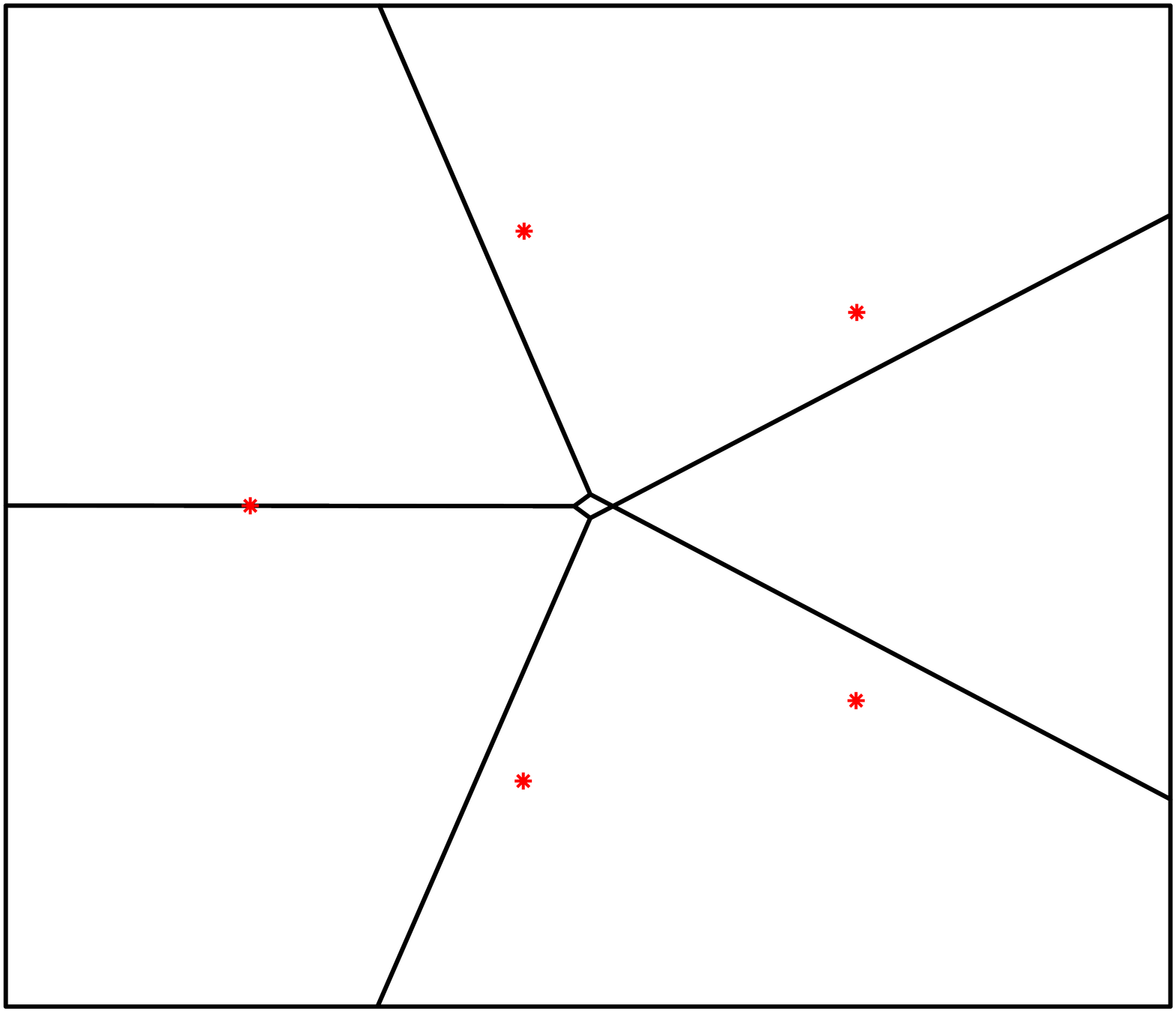} \label{fig:final-5}}%
\caption{Initial configuration (a), trajectories (b), and final configuration (c) for 5 agents running the \algoFull}%
\label{fig:sim-5}%
\end{figure*}

\begin{figure*}
\centering%
\subfloat[Performance]{\includegraphics[scale=0.16]{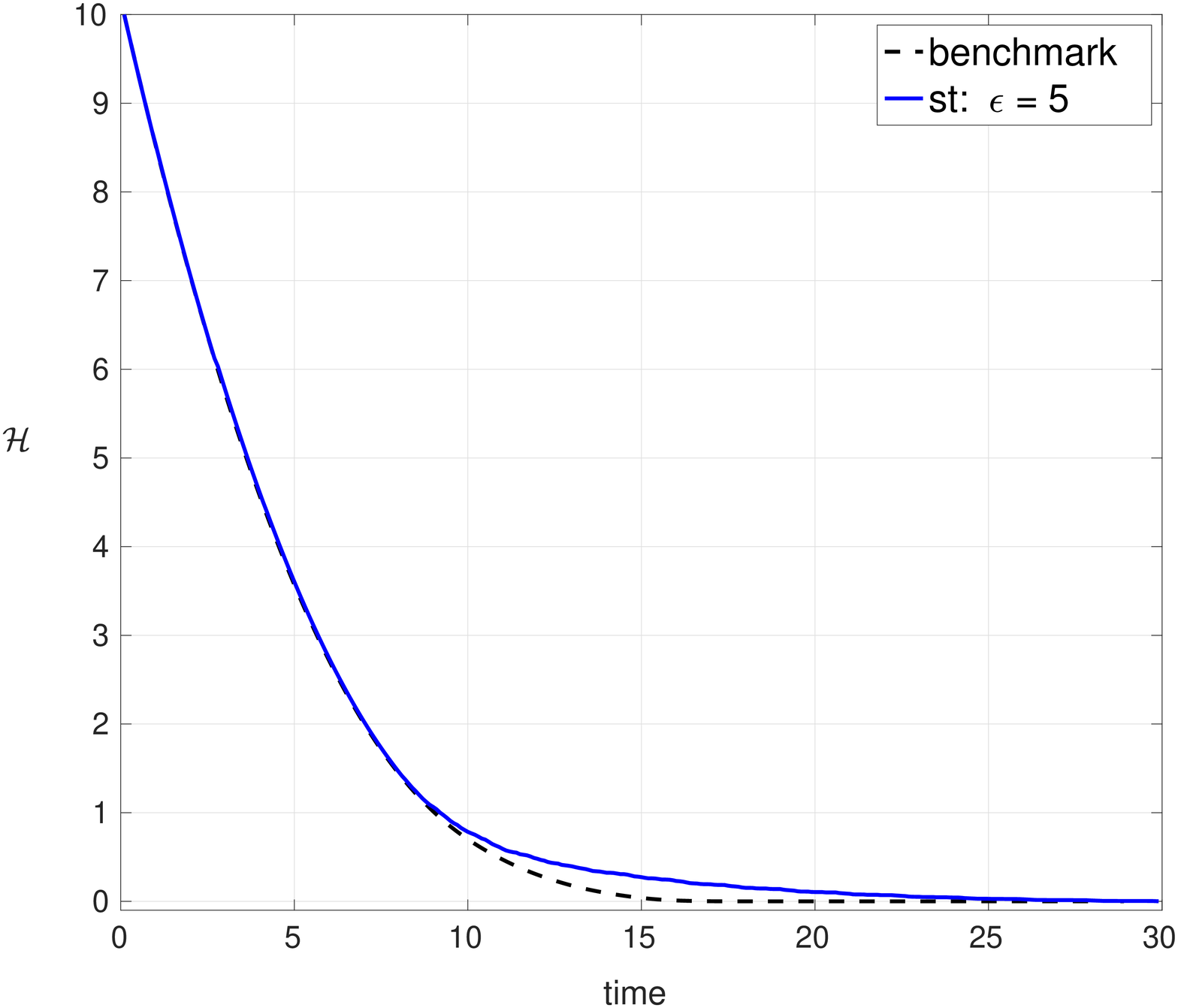} \label{fig:H}}%
~~\subfloat[Message count]{\includegraphics[scale=0.16]{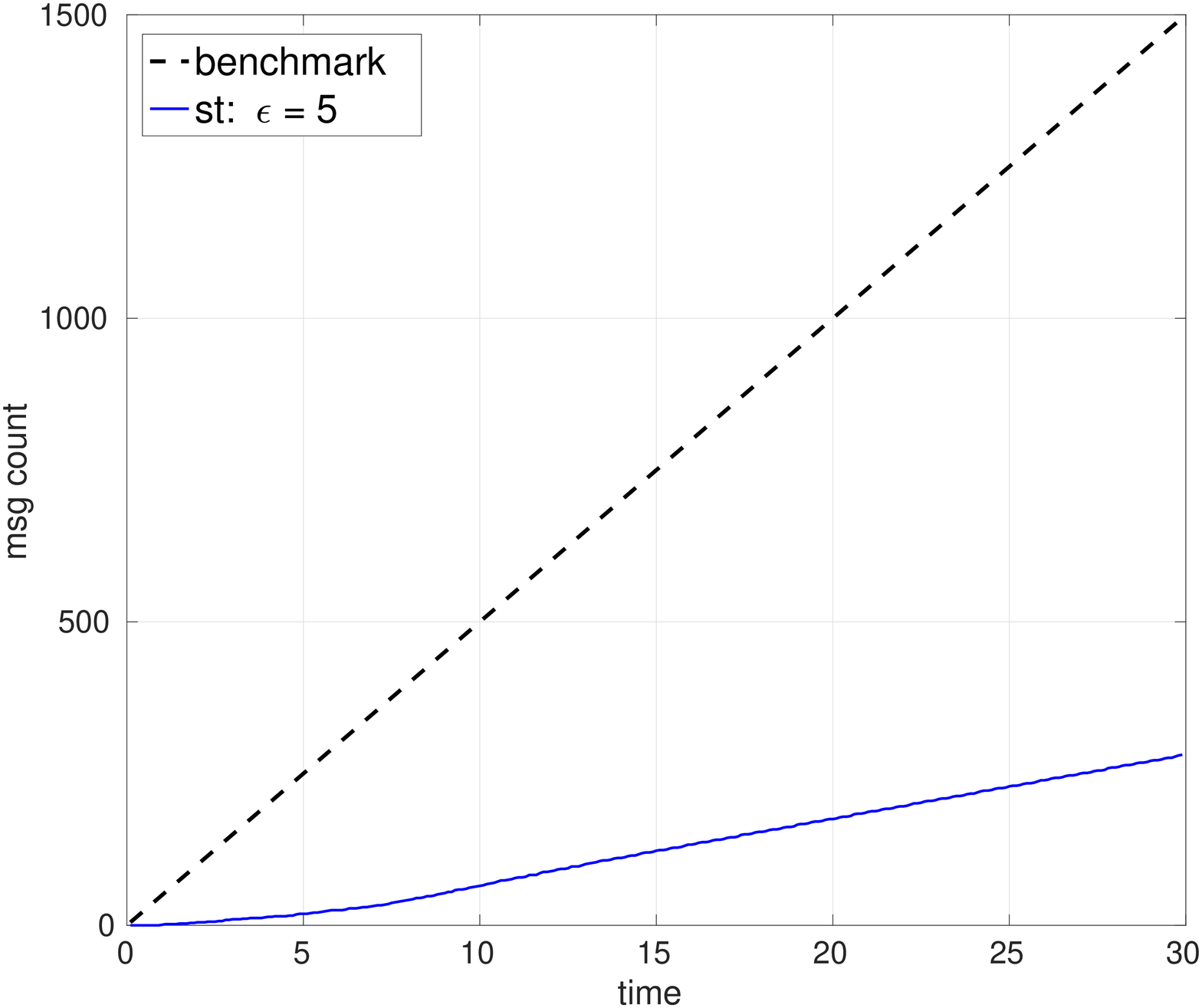} \label{fig:msg-cnt}}%
~~\subfloat[Power]{\includegraphics[scale=0.16]{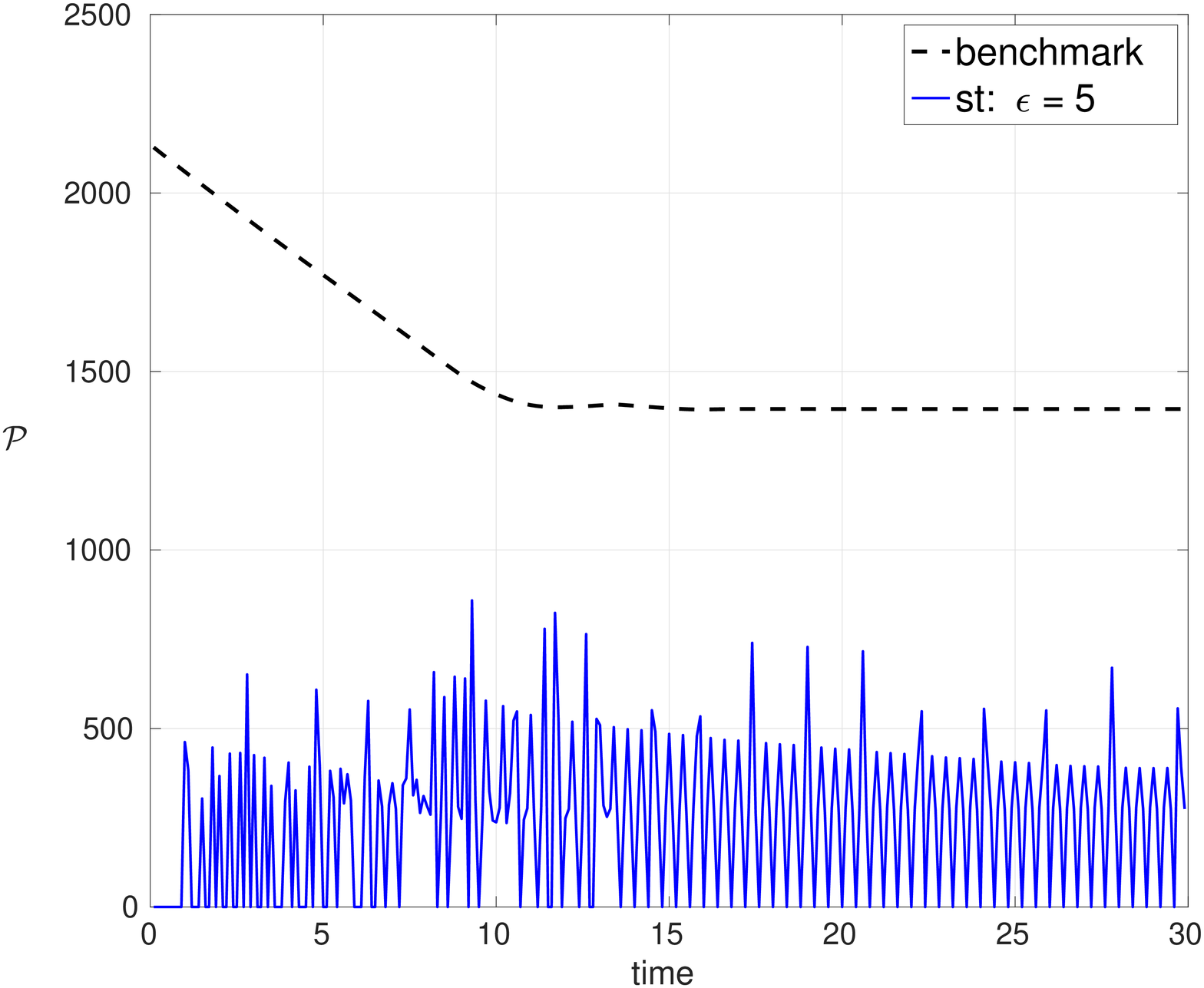} \label{fig:power}}%
\caption{Performance~(a), messages communicated between agents~(b), and power~(c) versus time.}%
\label{fig:results-5}%
\end{figure*}
\textbf{Proof of Proposition \ref{prop:main}}\\
Let $\gamma = \{\DD(t_\ell)\}_{t \in \integernonnegative}$ be an evolution of the self-triggered centroid algorithm. Define $\gamma' = \{\DD(t_\ell)\}_{t \in \integernonnegative}$ by $\DD'(t_\ell)=f_u(\DD(t_\ell))$. Note that $\loc{\DD(t_\ell)}=\loc{\DD'(t_\ell)}$. Since $\gamma'$ is a trajectory of $T$, lemma \ref{lm:weak-pos-invar} guarantees that $\Omega(\gamma')$ is weakly positively invariant and belongs to $\HH^{-1}(c)$ for some $c \in \real$. Next, it is shown that 
\begin{equation} \label{eq:omega-subset}
\begin{split}
\Omega(\gamma') \subset \setdef{\DD \in \dataspace}{i\in\ragents,
\TwoNorm{p_i^i-C_{\gW_i}} \leq \bound_i}
\end{split}
\end{equation}
We reason by contradiction. Assume there exists $\DD \in \Omega{\gamma}$ for which there is $i \in \ragents$ such that $\TwoNorm{p_i^i - C_{\gW_i}} > \bound_i$. By lemma \ref{lm:H-non-incr}, \ref{lm:motionconstraint} and the constraint given by (\ref{eq:motionconstraint}), any possible evolution from $\DD$ under $T$ will strictly decrease $\HH$. This is in contradiction with the fact that $\Omega(\gamma')$ is weakly positively invariant for $T$.

It is also noted that for each $i$ the inequality $\bound_i < \max{\{\TwoNorm{p_i^i - C_{\gW_i}},\epsilon\}}$ is satisfied at $\DD'(t_\ell))$, for all $\ell \in \integernonnegative$ an by continuity, this holds for $\Omega(\gamma')$ as well. That is,
\begin{equation}\label{eq:proof-max-bnd}
\bound_i < \max{\{\TwoNorm{p_i^i - C_{\gW_i}},\epsilon\}}
\end{equation}
for all $i \in \ragents$ and all $\DD \in \Omega(\gamma')$. Now it is shown that $\Omega(\gamma') \subset \setdef{\DD \in \dataspace}{i\in\ragents, p_i^i = C_{W_i}}$. Consider $\tilde{\DD} \in \Omega(\gamma')$. Since $\Omega(\gamma')$ is weakly positively invariant, there exists $\tilde{\DD}_1 \in \Omega(\gamma') \cap T(\tilde{\DD})$. Note that (\ref{eq:omega-subset}) implies that $\loc{\tilde{\DD}_1}=\loc{\tilde{\DD}}$ We consider two cases depending on whether agents have received information in $\tilde{\DD}_1$. If agent $i$ gets updated information then $\bound_i = 0$ and consequently from (\ref{eq:omega-subset}), $p_i^i = p'_i = C_{\gW_i} = C_{W_i}$ and the result follows. If agent~$i$ does not get updated information then $\bound(\tilde{\DD}_1^i) > \bound(\tilde{\DD}_1)$ and $\gW_i(\tilde{\DD}_1) \subset \gW_i(\tilde{\DD})$. Again using the fact that $\Omega(\gamma')$ is a weakly positively invariant set, there exist $\tilde{\DD}_2 \in \Omega(\gamma') \cap T(\tilde{\DD}_1)$ Reasoning repeatedly in this manner, the only case that needs to be discarded is when agent~$i$ never receives updated information. In this case $\TwoNorm{p_i^i - C_{\gW_i}} \rightarrow 0$ while $\bound_i$ monotonically increases towards $\diam(S)$. For sufficiently large $\ell$, $\TwoNorm{p_i^i - C_{\gW_i}} < \epsilon$. Then (\ref{eq:proof-max-bnd}) implies $\bound_i < \epsilon$, which contradicts the fact that $\bound_i$ tends towards $\diam(S)$.

\section{Simulations}\label{se:simulations}


\begin{figure*}
\centering%
\subfloat[Performance]{\includegraphics[scale=0.16]{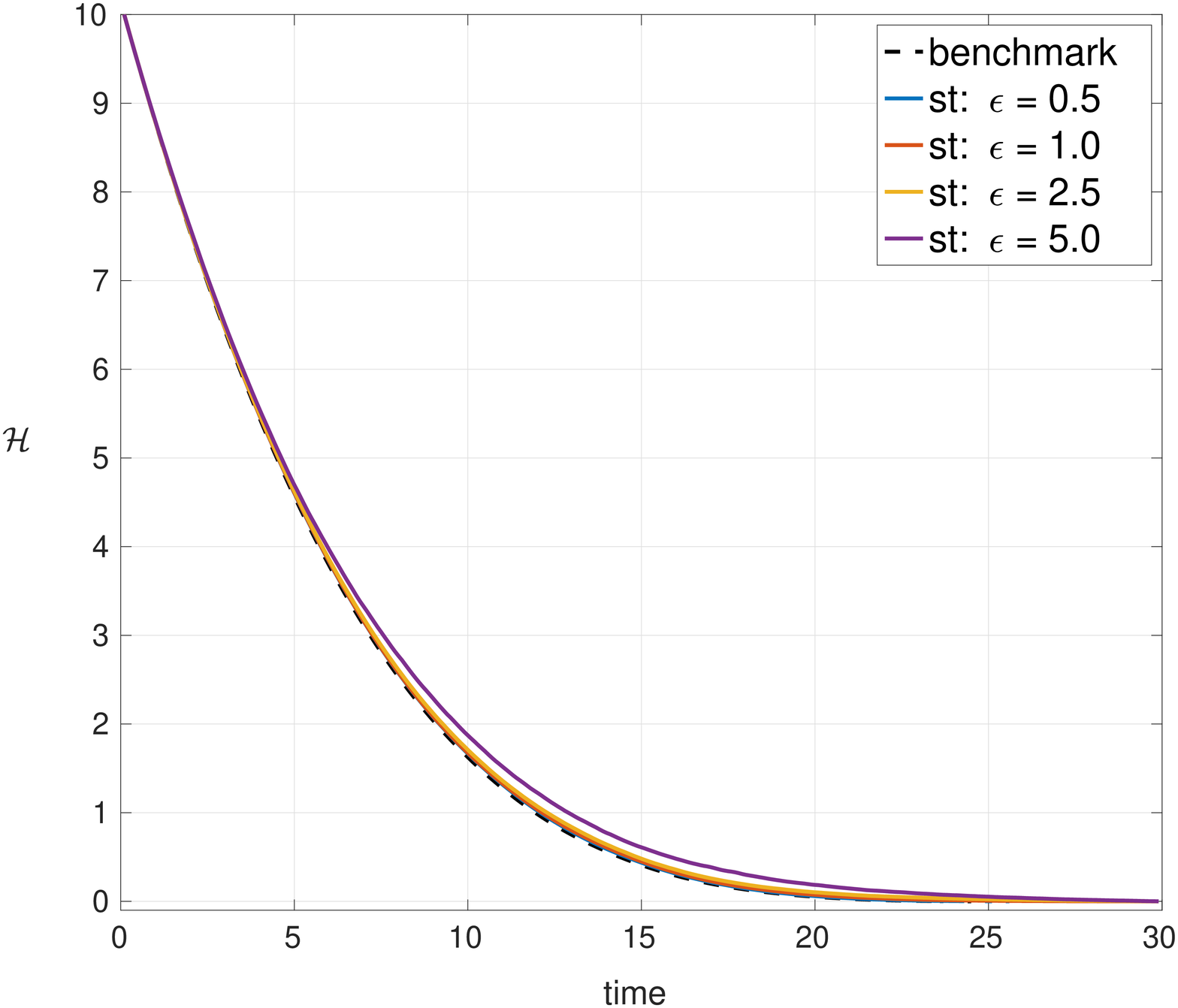} \label{fig:avg-H-5a}}%
\subfloat[Message count]{\includegraphics[scale=0.16]{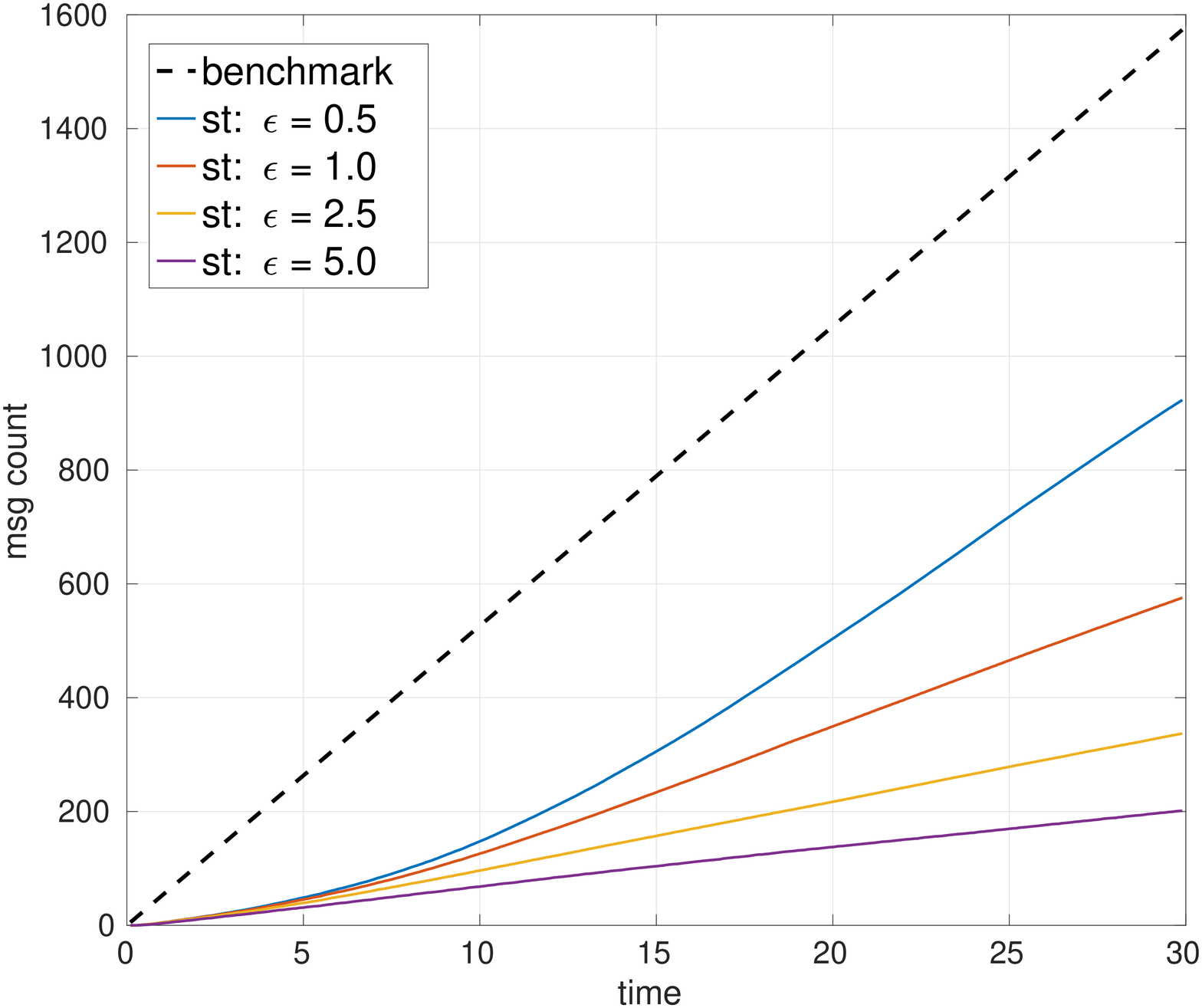} \label{fig:avg-M-5a}}%
\subfloat[Power]{\includegraphics[scale=0.16]{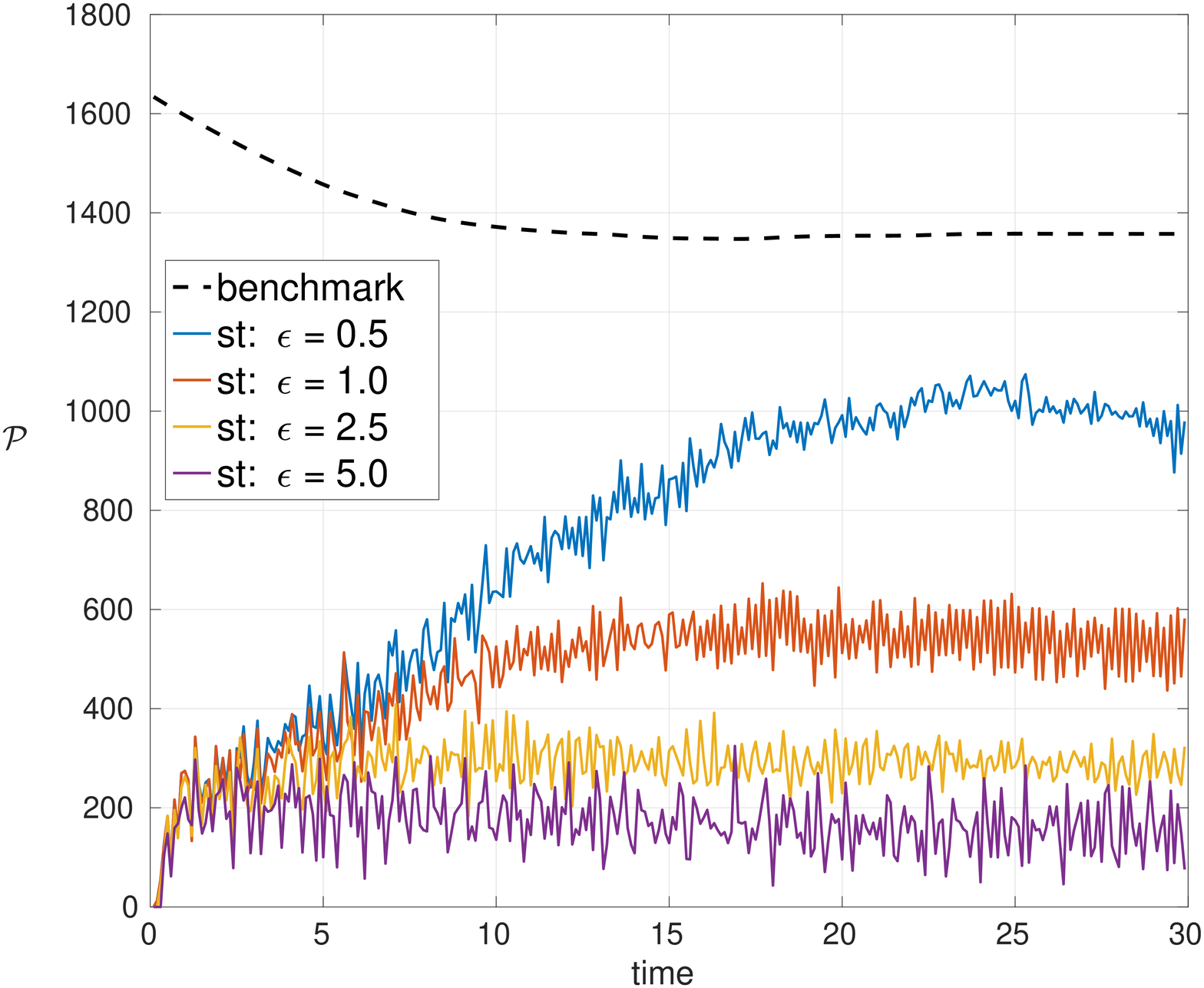} \label{fig:avg-P-5a}}%
\caption{Performance~(a), messages communicated between agents~(b), and power~(c) versus time. Average over 20 random initial configurations for different values of $\epsilon$.}
\label{fig:eps-results-5a}%
\end{figure*}

\begin{figure*}
\centering%
\subfloat[Performance]{\includegraphics[scale=0.16]{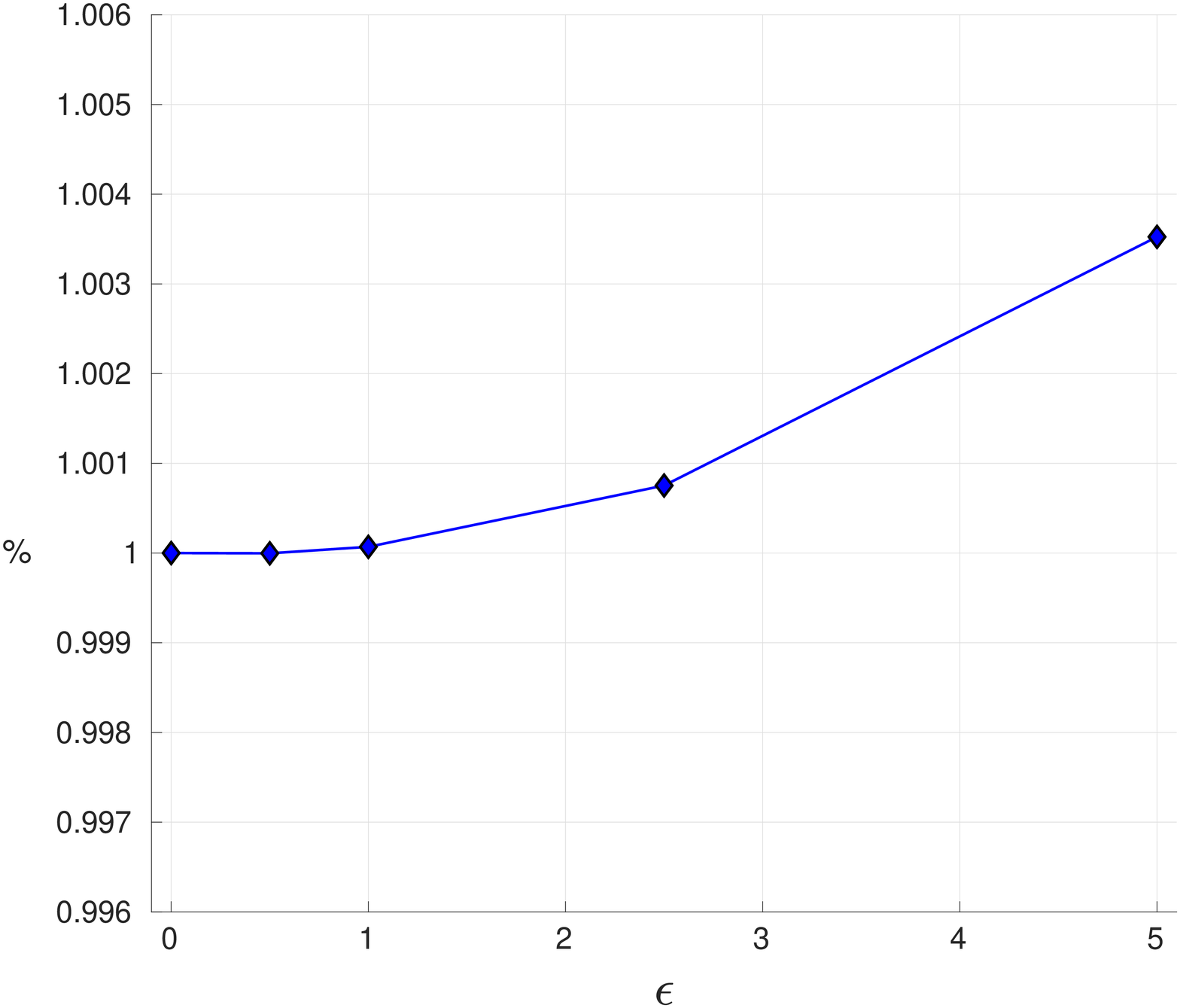} \label{fig:eps-H-5a}}%
\subfloat[Message count]{\includegraphics[scale=0.16]{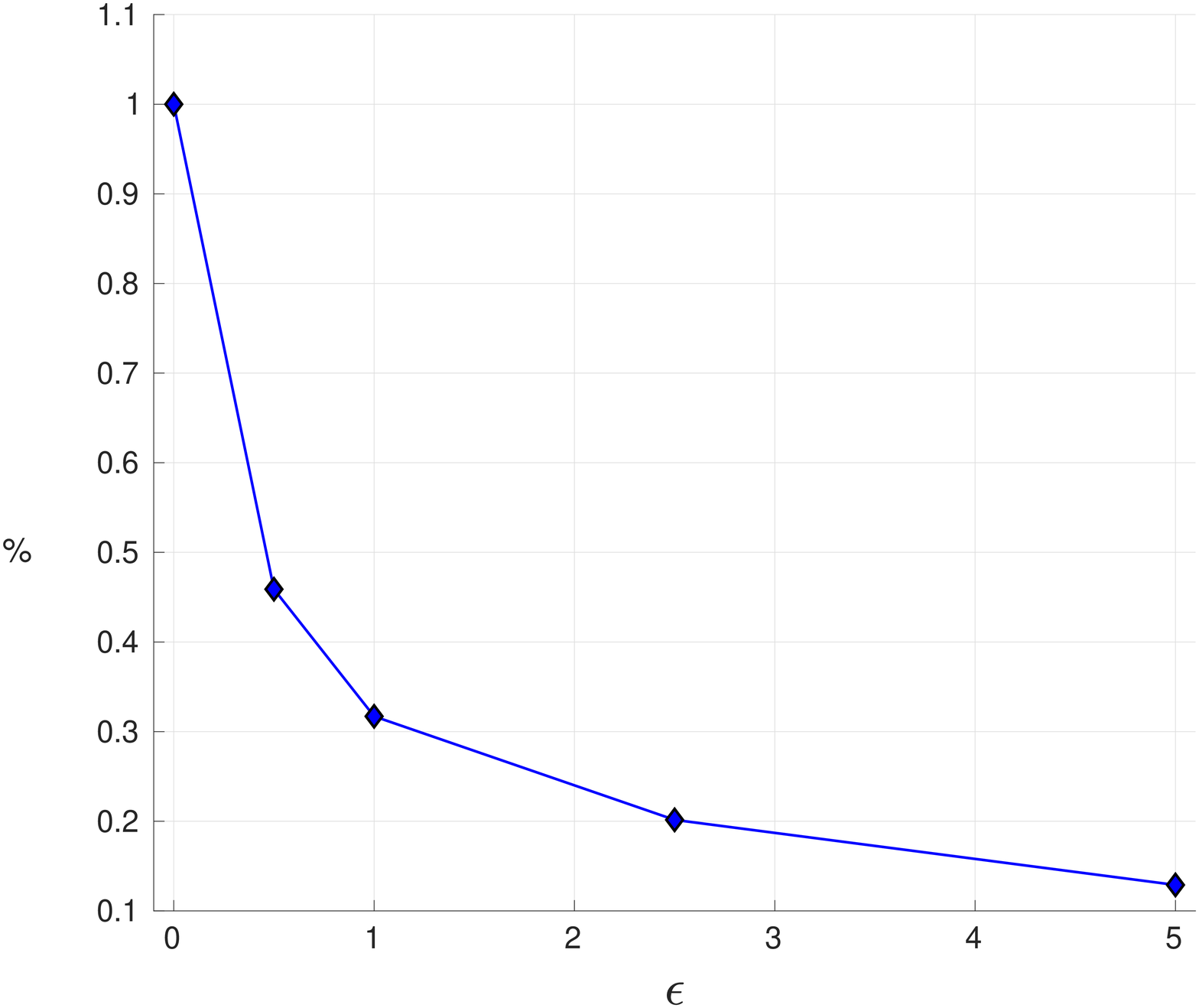} \label{fig:eps-M-5a}}%
\subfloat[Power]{\includegraphics[scale=0.16]{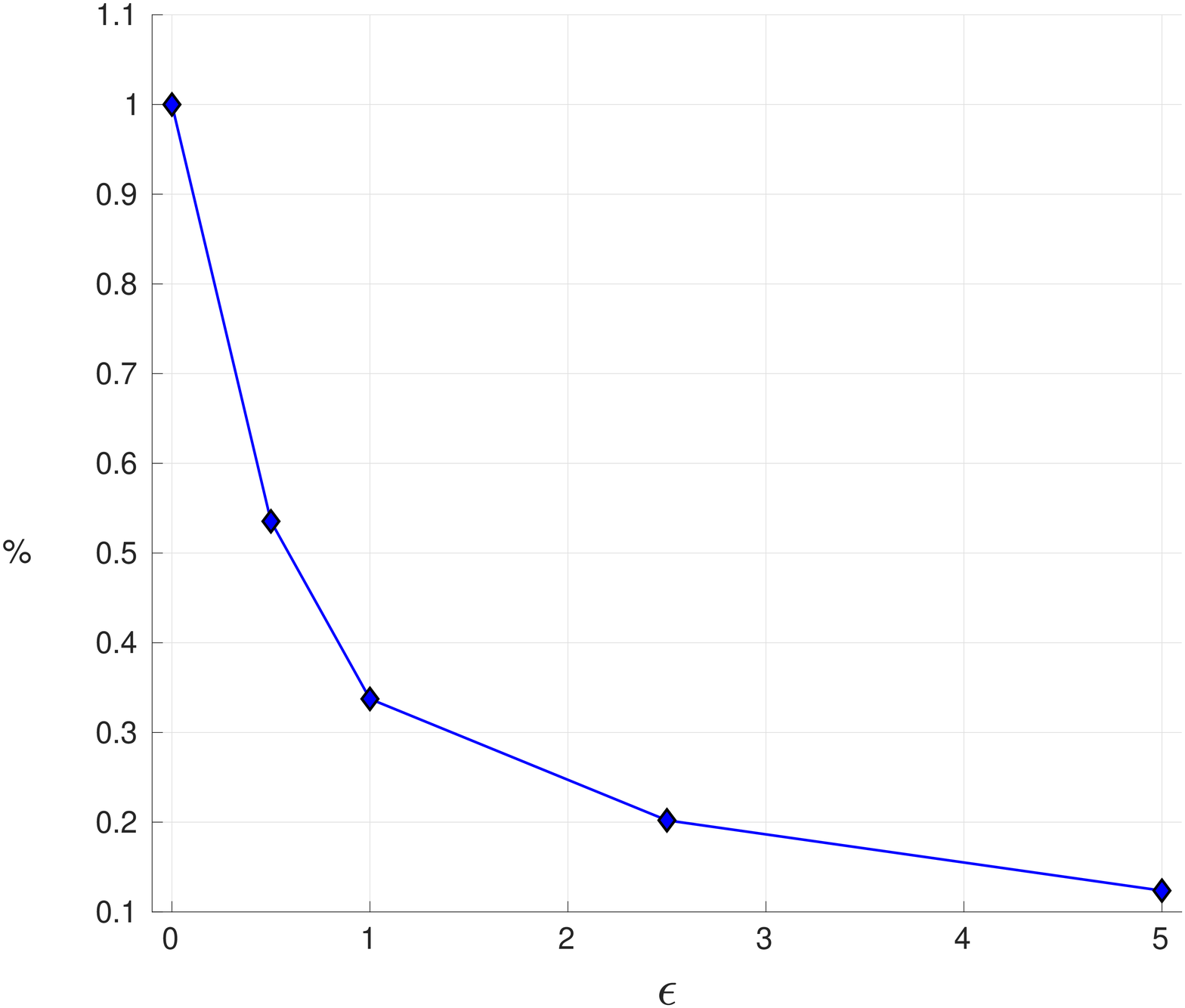} \label{fig:eps-P-5a}}%
\caption{Convergence of $\HH$~(a), total network message count~(b), and total network power~(c) averaged over 20 random initial configurations for each value of $\epsilon=(0,0.5,1,2.5,5)$ where $\epsilon=0$ corresponds to the benchmark case of continuous communication.}
\label{fig:eps-avg-5a}%
\end{figure*}

In this section, simulation results for the self-triggered deployment algorithm are presented. Simulations were performed with $n = 5$ agents moving in a $50 \text{m} \times 50 \text{m}$ area. The timestep was set to $\timestep = 0.1 \text{s}$ and all agents were given the same maximum velocity of $\vmax = 1$m/s. Multiple simulation iterations were performed by selecting different values of $\epsilon$ and generating random initial positions for agents on each iteration. Twenty iterations were carried out for each value of $\epsilon$. The values selected for $\epsilon$ were $\epsilon = \{0.5, 1.0, 2.5, 5.0\}$. To quantify the performance of the self-triggered method, the objective function $\HH$, the total transmission power, and the total number of messages transmitted were computed on every timestep. As in \cite{nowzari2012self}, the power output model for agent $i$ is given by
\[
\mathscr{P}_i = 10 \log_{10} \Bigg[ \sum_{j \in \ragents \setminus \{i\}} \beta ~10^{0.1 P_{i \rightarrow j} + \alpha \TwoNorm{p_i - p_j}} \Bigg]
\]
where $\alpha > 0$ and $\beta > 0$ are parameters that are dependent on the wireless medium and $P_{i \rightarrow j}$ is the power received from agent $i$ at agent $j$ in decibel-milliwatts. Simulation results were evaluated against a benchmark case that represents a centroidal continuous information update method where agents move toward their dominant cell centroid and positions are updated on every timestep $\timestep = 0.1 \text{s}$.

Figures \ref{fig:sim-5} and \ref{fig:results-5} display the results for the execution of a single simulation instance. Figures \ref{fig:sim-5} provides illustration of the initial configuration (\ref{fig:initial-5}), the trajectories traveled (\ref{fig:traj-5}), and the final configuration (\ref{fig:final-5}) of all agents following the self-triggered deployment strategy. Figure (\ref{fig:results-5}) shows a comparison against the benchmark case of the convergence of $\HH$ (\ref{fig:H}), the total message count (\ref{fig:msg-cnt}), and the communication power (\ref{fig:power}) at each timestep.
The results from figure \ref{fig:results-5} demonstrate how the self-triggered strategy can reduce both the total amount of communication and the power required to perform the deployment task. This is accomplished while still being capable of achieving convergence performance similar to that of a continuous or periodic communication strategy. Figures \ref{fig:eps-results-5a} and \ref{fig:eps-avg-5a} further illustrate this point by presenting results for combined values of $\epsilon$ where twenty random initial configurations for each $\epsilon$ are averaged together. In figure \ref{fig:eps-avg-5a}, the value $\epsilon=0$ corresponds to the benchmark case. These figures illustrate how varying $\epsilon$ affects the overall performance. It can be seen that the total message count and communication power decreases when the value of $\epsilon$ increases, while the the convergence rate of $\HH$ degrades. However, the convergence degradation of $\HH$ can be considered minimal when compared to the reduction in both message count and power. For the largest value $\epsilon=5$, the convergence of $\HH$ degrades by less than one-percent, while message count and communication power see a decrease of more than eighty-percent.

\section{Conclusions}\label{se:conclusions}
This paper presented a \algoFull for optimal deployment of $k$-order coverage control scenarios. The presented strategy combined an information update policy with a motion control law. The information update policy provided a method to determine when each agent should communicate with other agents in the network. Agents communicate in order to update their data storage. The decision to communicate is based on whether an agent can continue to contribute positively to the deployment objective. The motion control law provided a method for agents to move when the locations of other agents is uncertain due to the lack of communication. Through analysis, the proposed strategy was shown to provide guaranteed asymptotic convergence. The results have shown convergence similar to that of continuous and periodic position update methods. Simulation results were able to demonstrate the potential benefits of the proposed method by illustrating the ability of the \algoFull to not only reduce the amount of communication necessary to achieve the deployment goal, but also reducing the power consumed from communication.

\bibliographystyle{IEEEtran}
\bibliography{references}

\end{document}